\documentclass[USenglish]{lmcs}
\pdfoutput=1

\usepackage{lastpage}
\lmcsdoi{16}{1}{24}
\lmcsheading{}{\pageref{LastPage}}{}{}%
{Mar.~01,~2018}{Feb.~21,~2020}{}

\usepackage[utf8]{inputenc}
\usepackage[T1]{fontenc}
\usepackage{datetime}

\usepackage{amsmath}
\usepackage{mytikz}
\usepackage{amssymb}
\usepackage{stmaryrd}
\usepackage[multiuser,layout=margin]{fixme}
\usepackage{babel}
\usepackage{xspace}
\usepackage{hyperref}
\usepackage{url}
\usepackage{microtype}
\usepackage{multirow,tabulary}

\input{define.tex}

\FXRegisterAuthor{oc}{aoc}{OC}
\FXRegisterAuthor{mc}{amc}{MC}
\FXRegisterAuthor{cp}{acp}{CP}
\revision$LastChangedRevision: 1620 $


\title[Continuity of Functional Transducers]{Continuity of Functional
  Transducers:\texorpdfstring{\\}{} A Profinite Study of Rational Functions}

\begin{document}

\titlecomment{Extended version of the paper ``Continuity and
  Rational Functions'' appearing in \emph{Proceedings ICALP'17}.}

\author[M. Cadilhac]{Michaël Cadilhac\rsuper{a}}
\address{\lsuper{a}DePaul University}
\email{michael@cadilhac.name}

\author[O. Carton]{Olivier Carton\rsuper{b}}
\address{\lsuper{b}IRIF, Université Paris Diderot}
\email{olivier.carton@irif.fr}

\author[C. Paperman]{Charles Paperman\rsuper{c}}
\address{\lsuper{c}Université de Lille}
\email{charles.paperman@univ-lille.fr}

\begin{abstract}
  A word-to-word function is continuous for a class of languages~$\cV$ if its
  inverse maps $\cV$_languages to~$\cV$.  This notion provides a basis for an
  algebraic study of transducers, and was integral to the characterization of
  the sequential transducers computable in some circuit complexity classes.

  Here, we report on the decidability of continuity for functional transducers
  and some standard classes of regular languages.  To this end, we develop a
  robust theory rooted in the standard profinite analysis of regular languages.

  Since previous algebraic studies of transducers have focused on the sole
  structure of the underlying input automaton, we also compare the two algebraic
  approaches.  We focus on two questions: When are the automaton structure and
  the continuity properties related, and when does continuity propagate to
  superclasses?
\end{abstract}

\maketitle

\setcounter{tocdepth}{2}

\section{Introduction}

The algebraic theory of regular languages is tightly interwoven with fundamental
questions about the computing power of Boolean circuits and logics.  The most
famous of these braids revolves around $\vA$, the class of \emph{aperiodic} or
\emph{counter-free} languages.  Not only is it expressed using the logic
$\text{FO}[<]$, but it can be seen as the basic building block of $\ACz$, the
class of languages recognized by circuit families of polynomial size and
constant depth.  This class is in turn expressed by the logic
$\text{FO}[\text{arb}]$ (see~\cite{straubing94} for a lovely account).  This
pervasive interaction naturally suggests lifting this study to the functional
level, hence to \emph{rational functions}.  This was started
in~\cite{cadilhac-krebs-ludwig-paperman15}, where it was shown that a
subsequential (i.e., input-deterministic) transducer computes an $\ACz$ function
iff it preserves the regular languages of $\ACz$ by inverse image.  Buoyed by
this clean, semantic characterization, we wish to further investigate this
latter property for different classes: say that a function $f\colon A^* \to B^*$
is $\cV$_continuous, for a class of languages $\cV$, if for every language
$L \subseteq B^*$ of $\cV$, the language $f^{-1}(L)$ is also a language of
$\cV$.  Our main focus will be on deciding $\cV$_continuity for rational
functions; before listing our main results, we emphasize two additional
motivations.

First, there has been some historical progression towards this goal.  Noting,
in~\cite{pin-sakarovitch82}, that inverse rational functions provide a uniform
and compelling view of a wealth of natural operations on regular languages, Pin
and Sakarovitch initiated in~\cite{pin-sakarovitch85} a study of
regular-continuous functions.  It was already known at the time, by a result of
Choffrut~(see~\cite[Theorem 2.7]{berstel79}), that regular-continuity together
with some uniform continuity property \emph{characterize} functions computed by
subsequential transducers.  This characterization was instrumental in the study
of \reutschu~\cite{reutenaeur-schutzenberger95}, who already noticed the
peculiar link between uniform continuity for some distances on words and
continuity for certain classes of languages.  This link was tightened by Pin and
Silva~\cite{pin-silva05} who formalized a topological approach and generalized
it to rational relations.  More recently~\cite{pin-silva11}, the same authors
made precise the link unveiled by \reutschu, and developed a fascinating and
robust framework in which language continuity has a topological interpretation
(see the beginning of Section~\ref{sec:contapp}, as we build upon this theory).
Pin~and~Silva~\cite{pin-silva17} notably proposed thereafter a study of
functions for which continuity for a class is propagated to subclasses.  In
addition, Daviaud et al.\ \cite{DaviaudRT16,DaviaudJRV17} recently explored
continuity notions in the spirit of Choffrut's characterization to study
weighted automata and cost-register automata.

Second, the interweaving between languages, circuits, and logic that was alluded
to previously can in fact be formally stated (see
again~\cite{straubing94,tesson-therien07}).  A central property towards this
formalization is the correspondence between ``cascade products'' of automata,
stacking of circuits, and nesting of formulas, respectively.  Strikingly, these
operations can all be seen as inverse rational
functions~\cite{tesson-therien07}.  These operations are intrinsic in the
construction of complex objects: languages, circuits, and formulas are often
given as a sequence of simple objects to be composed (see, e.g.,~\cite[Section
5.5]{schneider04}).  We remark that a sufficient condition for the result of the
composition to be in some given class (of languages, circuits, or logic
formulas), is that each rational function be continuous for that class.  Hence
deciding continuity allows to give a sufficient condition for this membership
question \emph{without} computing the result of the composition, which is
subject to combinatorial blowup.

Here, we report on three questions, the first two relating continuity to the
other main algebraic approach to transducers, while allowing a more gentle
introduction to the evaluation of \emph{profinite words} by transducers:
\begin{itemize}
\item When does the transducer \emph{structure} (i.e., its so-called
  \emph{transition monoid}) impact its continuity?  The results of
  \reutschu~\cite{reutenaeur-schutzenberger95} can indeed be seen as the
  starting point of two distinct algebraic theories for rational functions; on
  the one hand the study of continuity, and on the other the study of the
  transition monoid of the transducer (disregarding the output).  This latter
  avenue was explored by~\cite{filiot-gauwin-lhote16}.
  We show in Section~\ref{sec:struct}:
  \begin{thm}
    Let $\cV$ be a variety of languages among $\vJ$, $\vR$, $\vL$, $\vDA$,
    $\vA$, $\vCom$, $\vAb$, $\vGnil$, $\vGsol$, or $\vG$.
    \begin{itemize}
    \item The statement ``Any rational function \emph{structurally} in $\cV$ is
      continuous for $\cV$'' holds for $\cV \in \{\vA, \vGsol, \vG\}$ and does
      not otherwise;
    \item The statement ``Any rational function continuous for $\cV$ is
      \emph{structurally} in $\cV$'' holds for $\cV \in \{\vGnil, \vGsol, \vG\}$
      and does not otherwise.
    \end{itemize}
  \end{thm}
\item What is the impact of \emph{variety inclusion} on the inclusion of the
  related classes of continuous rational functions?  When focusing on transducer
  structure alone, there is a natural propagation to superclasses; when is it
  the case for continuity?  We show in Section~\ref{sec:inc}:
  \begin{thm}
    Let $\cV$ and $\cW$ be two different varieties of languages among $\vJ$,
    $\vR$, $\vL$, $\vDA$, $\vA$, $\vCom$, $\vAb$, $\vGnil$, $\vGsol$, or $\vG$.
    The statement ``all rational functions continuous for $\cV$ are continuous
    for $\cW$'' holds only when one of these properties is satisfied:
    \begin{itemize}
    \item $\cV, \cW \in \{\vGnil, \vGsol, \vG\}$ and $\cV \subseteq \cW$;
    \item $\cV = \vAb$ and $\cW = \vCom$;
    \item $\cV = \vDA$ and $\cW = \vA$.
    \end{itemize}
  \end{thm}
\item When is $\cV$_continuity decidable for rational functions?  We show in
  Section~\ref{sec:dec}:
  \begin{thm}
    Let $\cV$ be a variety of languages among $\vJ$, $\vR$, $\vL$, $\vDA$,
    $\vA$, $\vCom$, $\vAb$, $\vGsol$, or $\vG$.  It is decidable, given an
    unambiguous rational transducer, whether it realizes a function continuous
    for $\cV$.
  \end{thm}
  \noindent%
  This constitutes our main contribution; note that the case $\vGnil$ is left
  open.
\end{itemize}

\tableofcontents

\section{Preliminaries}

We assume some familiarity with the theory of automata and transducers, and
concepts related to metric spaces (see, e.g.,~\cite{berstel79,pin} for
presentations pertaining to our topic).  We first settle the notation for these
prerequisites.

We will use $A$ and $B$ for alphabets, and $A^*$ for words over~$A$, with $1$
the empty word.  For each word $u$, there is a smallest $v$, called the
\emph{primitive root} of $u$, such that $u=v^c$ for some $c$; if $c=1$, then $u$
is itself \emph{primitive}.  We write $|u|$ for the length of a word $u \in A^*$
and $\cts(u)$ for the set of letters that appear in $u$.

Let \(L \subseteq A^*\) be a language.  We write $L\comp$ for the complement of
$L$.  For a word $u \in A^*$, we write $u^{-1}L$ for
$\{v \mid u\cdot v \in L\}$, and symmetrically for $Lu^{-1}$, these two
operations being called the left and right quotients of $L$ by $u$,
respectively.  We naturally extend concatenation and quotients to binary
relations, in a component-wise fashion, e.g., for $R \subseteq A^* \times A^*$
and a pair $\rho \in A^* \times A^*$, we may use $\rho^{-1}R$ and $R\rho^{-1}$.
More generally, with \(\Id\) the identity relation, we will write, e.g.,
\(\Id \cdot ((x^*, x^*)\rho^{-1})\) for the pairs \((ww_1, ww_2)\) of words such
that \(w_i \in x^*\rho_i^{-1}\), \(i = 1,2\).

A \emph{variety} is a mapping $\cV$ which associates with each alphabet $A$ a
set $\cV(A^*)$ of regular languages closed under the Boolean operations and
quotient, and such that for any morphism $h\colon A^* \to B^*$ and any
$L \in \cV(B^*)$, we have that $h^{-1}(L) \in \cV(A^*)$.  Reg is the variety
that maps every alphabet $A$ to the set $\text{Reg}(A^*)$ of regular languages
over $A$.

Given two languages $K, L \subseteq A^*$, we say that they are |$\cV$_separable|
if there is a $S \in \cV(A^*)$ such that $K \subseteq S$ and
$L \cap S = \emptyset$.  Since \(\cV\) is closed under complement, \(K, L\) are
\(\cV\)_separable iff \(L, K\) are.  Naturally, \(L \in \cV(A^*)\) iff \(L\) is
\(\cV\)_separable from \(L\comp\).

\subsubsection*{Transducers.}  A transducer $\tau$ is a 9-tuple
$(Q, A, B, \delta, I, F, \lambda, \mu, \rho)$ where $(Q, A, \delta, I, F)$ forms
a nondeterministic automaton (i.e., $Q$ is a state set, $A$ an input alphabet,
$\delta \subseteq Q \times A \times Q$ a transition set, $I \subseteq Q$ a set
of initial states, and $F \subseteq Q$ a set of final states), and additionally,
$B$ is an output alphabet and
$\lambda\colon I \to B^*, \mu\colon \delta \to B^*, \rho\colon F\to B^*$ are the
output functions.  We write $\Tif{\tau, q, q'}$ for $\tau$ with $I := \{q\}$ and
$F := \{q'\}$, adjusting $\lambda$ and $\rho$ to output $1$ if they were
undefined on these states.  Similarly, $\Tif{\tau, q, \nop}$ is $\tau$ with
$I := \{q\}$ and $F$ unchanged, and symmetrically for~$\Tif{\tau, \nop, q}$.
For $q \in Q$ and $u \in A^*$, we write $q.u$ for the set of states reached from
$q$ by reading $u$.  We assume that all the transducers and automata under study
have no useless state, that is, all states appear in some accepting path.

With $w \in A^*$, let $t_1t_2\cdots t_{|w|} \in \delta^*$ be an accepting path
for $w$, starting in a state $q \in I$ and ending in some $q' \in F$.  The
output of this path is $\lambda(q)\mu(t_1)\mu(t_2)\cdots\mu(t_n)\rho(q')$, and
we write $\tau(w)$ for the set of outputs of such paths.  We use $\tau$ for both
the transducer and its associated partial function from $A^*$ to subsets of
$B^*$.  Relations of the form $\{(u, v) \mid v \in \tau(u)\}$ are called
\emph{rational relations}.

The transducer $\tau$ is \emph{unambiguous} if there is at most one accepting
path for each word.  In that case $\Tif{\tau, q, q'}$ is also an unambiguous
transducer for any states $q, q'$.  When $\tau$ is unambiguous, it realizes a
(partial) word-to-word function: the set of functions computed by unambiguous
transducers is the set of \emph{rational functions}.  Further restricting, if
the underlying automaton is deterministic, we say that $\tau$ is a
\emph{subsequential} transducer.  If $\tau$ is a finite union of subsequential
transducers of disjoint domains, we say that $\tau$ is
\emph{plurisubsequential}.

\subsubsection*{Word distances, profinite words.}  For a variety $\cV$ of regular
languages, we define a distance between words for which, intuitively, two words
are close if it is hard to separate them with $\cV$ languages.  Define
$d_\cV(u, v)$, for words $u, v \in A^*$, to be $2^{-r}$ where $r$ is the size of
the smallest automaton that recognizes a language of $\cV(A^*)$ that separates
$\{u\}$ from $\{v\}$; if no such language exists, then $d_\cV(u, v) = 0$.  It can
be shown that this distance is a \emph{pseudo-ultrametric}~\cite[Section
VII.2]{pin}; we make only implicit and innocuous use of this fact.

The complete metric space that is the completion of $(A^*, d_{\text{Reg}})$ is
denoted $\pro{A^*}$ and is called the \emph{free profinite monoid}, its elements
being the \emph{profinite words}, and the concatenation being naturally
extended.  By definition, if ${(u_n)}_{n>0}$ is a Cauchy sequence, it should hold
that for any regular language $L$, there is a $N$ such that either all $u_n$
with $n > N$ belong to $L$, or none does.  For any $x \in A^*$, define the
profinite word $x^\omega = \lim x^{n!}$, and more generally, for any \(c > 0\),
$x^{\omega-c} = \lim x^{n!-c}$.  That ${(x^{n!})}_{n>0}$ is a Cauchy sequence is a
starting point of the profinite theory~\cite[Proposition~VI.2.10]{pin}; it is
also easily checked that $x^{c \times \omega} = \lim x^{c\times n!}$ is equal to
$x^\omega$ for any integer $c \geq 1$.  Given a language $L \subseteq A^*$, we
write $\clos{L} \subseteq \pro{A^*}$ for its closure, and we note that if $L$ is
regular, $\clos{L}\comp = \clos{L\comp}$---the complement being taken in
\(\pro{A^*}\) in the left-hand side and in \(A^*\) in the right-hand side.
Furthermore, for $L'$ regular, $\clos{L \cup L'} = \clos{L} \cup \clos{L'}$, and
similarly for intersection (see~\cite[Theorem VI.3.15]{pin}).

\subsubsection*{Equations.}  For $u, v \in \pro{A^*}$, a language $L \subseteq A^*$
\emph{satisfies the (profinite) equation} $u = v$ if for any words
$s, t \in A^*$,
$[s\cdot u\cdot t \in \clos{L} \Leftrightarrow s \cdot v \cdot t \in \clos{L}]$.
Similarly, a class of languages satisfies an equation if all the languages of
the class satisfy it.  For a variety $\cV$, we write $u =_\cV v$, and say that
$u$ is equal to $v$ in $\cV$, if $\cV(A^*)$ satisfies $u = v$.  For a partial
function $f$, $f(u) =_\cV f(v)$ means that either both $f(u)$ and $f(v)$ are
undefined, or they are both defined and equal in $\cV$.

Given a set $E$ of equations over $\pro{A^*}$, the class of languages |defined|
by $E$ is the class of languages over $A^*$ that satisfy all the equations of
$E$.  Reiterman's theorem shows in particular that for any variety $\cV$ and any
alphabet $A$, $\cV(A^*)$ is defined by a set of equations (the precise form of
which being studied in~\cite{gehrke-grigorieff-pin08}).

\subsubsection*{More on varieties.}  Borrowing from Almeida and
Costa~\cite{almeida-costa17}, we say that a variety $\cV$ is
\emph{supercancellative} when for any alphabet $A$, any $u, v \in \pro{A^*}$ and
$x, y \in A$, if $u\cdot x =_\cV v \cdot y$ or $x \cdot u =_\cV y \cdot v$, then
$u =_\cV v$ and $x = y$.  This implies in particular that for any word
$w \in A^*$, both $w\cdot A^*$ and $A^* \cdot w$ are in $\cV(A^*)$.  We further
say that a variety $\cV$ \emph{separates words} if for any $s, t \in A^*$,
$\{s\}$ and $\{t\}$ are $\cV$_separable.

Our main applications revolve around some classical varieties, that we define
over any possible alphabet $A$ as follows, where $x, y$ range over all of $A^*$,
and $a, b$ over $A$:

\vspace{1em}
\def\dby{def.~by~}%
\noindent\begin{minipage}{.49\textwidth}\small\begin{itemize}
  \item $\vJ$, \dby ${(xy)}^\omega \cdot x = y \cdot {(xy)}^\omega = {(xy)}^\omega$
  \item $\vR$, \dby ${(xy)}^\omega \cdot x = {(xy)}^\omega$
  \item $\vL$, \dby $y \cdot {(xy)}^\omega = {(xy)}^\omega$
  \item $\vDA$, \dby $x^\omega\cdot z \cdot x^\omega = x^\omega$ for all
    $z \in {\cts(x)}^*$
  \item $\vA$, \dby $x^{\omega +1} = x^\omega$
  \end{itemize}
\end{minipage}\;
\rule[-1cm]{.1pt}{2cm}\;\;%
\begin{minipage}{.45\linewidth}
  \small
  \begin{itemize}[nosep]
  \item $\vCom$, \dby $ab = ba$
  \item $\vAb$, \dby $ab = ba$ and $a^\omega = 1$
  \item $\vGnil$, the languages rec.~by nilpotent groups
  \item $\vGsol$, the languages rec.~by solvable groups
  \item $\vG$, the languages rec.~by groups
  \end{itemize}
\end{minipage}
\vspace{.7em}

\noindent
The varieties included in $\vA$ are called \emph{aperiodic varieties} and those
in $\vG$ are called \emph{group varieties}.  Precise definitions, in particular
for the group varieties, can be found in~\cite{straubing94,pin-weil96}; we
simply note that in group varieties, $x^\omega$ equals $1$ for all $x \in A^*$.
All these varieties except for $\vAb$ and $\vCom$ separate words, and only
$\vDA$ and $\vA$ are supercancellative.  They satisfy:

\vspace{.7em}
\centerline{%
  \begin{tikzpicture}[baseline={(a.east)}]
  \node (a) {$\vJ=\vR\cap\vL$};
    \node [anchor=west, above right=0.1 of a] {$\vR$};
    \node at ($(a)+(0.75,0.5)$) [,rotate=35] {$\subsetneq$};
    \node at ($(a)+(0.75,-0.5)$) [,rotate=-35] {$\subsetneq$};
    \node [anchor=west, below right=0.1 of a] {$\vL$};
    \node at ($(a)+(1.85,-0.5)$) [,rotate=35] {$\subsetneq$};
    \node at ($(a)+(1.85,0.5)$) [,rotate=-35] {$\subsetneq$};
    \node at ($(a)+(2.9,0)$) [] {$\subsetneq$};
    \node [anchor=west, right=0.85 of a] {$\vDA$};
    \node [anchor=west, right=2 of a] {$\vA$};
  \end{tikzpicture}
  \qquad
  \begin{tikzpicture}[baseline=(a.east)]
    \node (a) {$\vAb=\vG \cap \vCom\subsetneq\vGnil\subsetneq\vGsol\subsetneq \vG$};
    \node at ($(a.north west)+(2.2,0.4)$) {$\vCom$};
    \node at ($(a.north west)+(1.5,0.25)$)[rotate=35] {$\subsetneq$};
  \end{tikzpicture}
}

\subsubsection*{On transducers and profinite words.}  For a profinite word $u$ and a
state $q$ of an unambiguous transducer $\tau$, the set $q.u$ is well defined;
indeed, with $u = \lim u_n$, the set $q.u_n$ is eventually constant, as
otherwise for some state $q'$, the domain of $\Tif{\tau, q, q'}$ would be a
regular language that separates infinitely many $u_n$'s.

A transducer $\tau\colon A^* \to B^*$ is a \emph{$\cV$_transducer},\footnote{The
  usual definition of $\cV$_transducer is based on the so-called transition
  monoid of~$\tau$, see, e.g.,~\cite{reutenaeur-schutzenberger95}; the
  definition here is easily seen to be equivalent
  by~\cite[Lemma~3.2]{almeida99}
  and~\cite[Lemma~1]{cadilhac-krebs-ludwig-paperman15}.} for a variety $\cV$, if
for some set of equations~$E$ defining~$\cV(A^*)$, for all $(u = v) \in E$ and
all states $q$ of $\tau$, the equality $q.u = q.v$ holds.  A rational function is
\emph{$\cV$_realizable} if it is realizable by a $\cV$_transducer.

\subsubsection*{Continuity.}  For a variety $\cV$, a function
$f\colon A^* \to B^*$ is $\cV$_continuous\footnote{A note on terminology: There
  has been some fluctuation on the use of the term ``continuous'' in the
  literature, mostly when a possible incompatibility arises with topology.
  In~\cite{pin-silva17}, the authors use the term ``preserving'' in the more
  general context of functions from monoids to monoids.  In our study, we focus
  on word to word functions, in which the natural topological context provides a
  solid basis for the use of ``continuous,'' as used
  in~\cite{pin-silva05,cadilhac-krebs-ludwig-paperman15}.}  iff for any
$L \in \cV(B^*)$, $f^{-1}(L) \in \cV(A^*)$.  We mostly restrict our attention to
rational functions.  Since they are computed by transducers, there are countably
many such functions.  We note that many more Reg_continuous functions exist, in
particular uncomputable ones:
\begin{prop}
  There are uncountably many Reg_continuous functions.
\end{prop}
\begin{proof}
  Consider a strictly increasing function $g\colon \bbn \to \bbn$.  Define
  $f\colon {\{a\}}^* \to {\{a\}}^*$ by $f(a^n) = a^{g(n)!}$.  Recall that any
  regular language over a unary alphabet is a finite union of languages of the
  form $a^i{(a^j)}^*$.  Moreover, we have that $f^{-1}(a^i{(a^j)}^*)$ is finite when
  $i \not\equiv 0 \bmod j$, and cofinite otherwise, thus $f$ is Reg_continuous.
  There are however uncountably many increasing functions $g$, hence uncountably many
  Reg_continuous functions~$f$.
\end{proof}

Continuity is a formal notion of ``functions being compatible with a class of
languages.''  An equally valid notion could be to consider classes of functions
that contain the characteristic functions of the languages, and closed under
composition; it turns out that the largest such class coincides with the class
of continuous functions.  Indeed, writing $\chi_L\colon A^* \to \{0, 1\}$ for
the characteristic function of a language $L \subseteq A^*$:
\begin{prop}
  Let $\cV$ be a variety such that $\{1\} \in \cV({\{0, 1\}}^*)$.  Let
  $\cF$ be the \emph{largest} class of functions such that:
  \begin{enumerate}
  \item For any alphabet $A$,
    $\cF \cap {\{0, 1\}}^{A^*} = \{\chi_L \mid L \in \cV(A^*)\}$;
  \item $\cF$ is closed under composition.
  \end{enumerate}
  The class $\cF$ is well defined and it coincides with the class of
  $\cV$_continuous functions.
\end{prop}
\begin{proof}
  We say that a class of functions is \emph{good} if it satisfies
  properties (1) and (2).  We show that the $\cV$_continuous functions form
  a good class, and that any good class is included in the $\cV$_continuous
  functions.  This implies that there \emph{is} a largest good class, and that
  it coincides with the class of $\cV$_continuous functions, as claimed.

  \proofstep{Continuous functions form a good class}%
  Clearly, the class of $\cV$_continuous functions is closed under composition.
  Now consider a $\cV$_continuous function $f \colon A^* \to \{0, 1\}$.  By
  continuity, $L = f^{-1}(\{1\})$ is in $\cV(A^*)$, since by hypothesis
  $\{1\} \in \cV({\{0,1\}}^*)$.  Hence $f = \chi_L$ for some $L \in \cV(A^*)$,
  concluding this step.

  \proofstep{Functions in good classes are continuous}%
  Let $f\colon A^* \to B^*$ be in a good class, and let $L \in \cV(B^*)$; we
  ought to show that $f^{-1}(L)$ is in $\cV(A^*)$.  We have:
  \begin{align*}
    f^{-1}(L) & = f^{-1}(\chi_L^{-1}(1))\\
              & = {(\chi_L \circ f)}^{-1}(1)\\
              & = g^{-1}(1)\enspace. \tag{with $g = \chi_L \circ f$} 
  \end{align*}
  Note that $\chi_L$ is by hypothesis in the good class, and it being closed
  under composition, $g$ also belongs to the good class.  Since
  $g \in {\{0, 1\}}^{A^*}$, there is a $L' \in \cV(A^*)$ such that
  $g = \chi_{L'}$.  This implies that $f^{-1}(L) = L'$, and it thus belongs to
  $\cV(A^*)$.
\end{proof}

\section{Continuity: The profinite approach}\label{sec:contapp}

We build upon the work of Pin and Silva~\cite{pin-silva05} and develop tools
specialized to rational functions.  In Section~\ref{sec:preslem}, we present a
lemma asserting the equivalence between $\cV$_continuity and the
``preservation'' of the defining equations for $\cV$.  In the sections
thereafter, we specialize this approach to rational functions.  As noted
in~\cite{pin-silva05}, it often occurs that results about rational functions can
be readily applied to the larger class of Reg_continuous functions; here, this
is in particular the case for the Preservation Lemma of
Section~\ref{sec:preslem}.

The connection to the classical notion of continuity is given by the next
Theorem.
\begin{thmC}[{\cite[Theorem 4.1]{pin-silva11}}]\label{thm:ucont}
  Let $f\colon A^* \to B^*$.  It holds that $f$ is $\cV$_continuous iff $f$
  is uniformly continuous for the distance $d_\cV$.
\end{thmC}

Consequently, if $f$ is Reg_continuous then it has a unique continuous extension
to the free profinite monoid with domain \(\clos{f^{-1}(B^*)}\), written
$\ext{f}\colon \pro{A^*} \to \pro{B^*}$.
The salient property of this mapping is that it is continuous in the
\emph{topological sense} (see, e.g.,~\cite{pin}).  For our specific needs, we
simply mention that it implies that for any regular language $L$, we have that
$\ext{f}^{-1}(\clos{L})$ is closed (that is, it is the closure of some set).


\subsection{The Preservation Lemma: Continuity is equivalent to preserving
  equations}\label{sec:preslem}

The Preservation Lemma gives us a key characterization in our study: it ties
together continuity and some notion of preservation of equations.  This can be
seen as a generalization for functions of the notion of equation satisfaction
for languages.  We will need the following technical lemma that extends~\cite[Proposition~VI.3.17]{pin} from morphisms to arbitrary Reg_continuous
functions; interestingly, this relies on a quite different proof.
\begin{lem}\label{lem:pinext}
  Let $f\colon A^* \to B^*$ be a Reg_continuous function and $L$ a regular
  language.  The equality $\ext{f}^{-1}(\clos{L}) = \clos{f^{-1}(L)}$ holds.
\end{lem}
\begin{proof}
  First note that $f^{-1}(L) \subseteq \ext{f}^{-1}(\clos{L})$, and that the
  right-hand side of this inclusion is closed.  Hence
  $\clos{f^{-1}(L)} \subseteq \ext{f}^{-1}(\clos{L})$.

  For the converse inclusion, first write $D = f^{-1}(B^*)$, a regular language
  by hypothesis.  We have that
  $\ext{f}^{-1}(\clos{L}) = (\ext{f}^{-1}(\clos{L}\comp))\comp \cap \clos{D}$,
  and similarly, $f^{-1}(L) = (f^{-1}(L\comp))\comp \cap D$.  This latter
  equality implies that
  $\clos{f^{-1}(L)} = \clos{f^{-1}(L\comp)}\comp \cap \clos{D}$, since
  $f^{-1}(L\comp)$ and $D$ are regular.

  Hence the inclusion to be shown, that is,
  $\ext{f}^{-1}(\clos{L}) \subseteq \clos{f^{-1}(L)}$, is equivalent to:
  \begin{align*}
    (\ext{f}^{-1}(\clos{L}\comp))\comp \cap \clos{D} \subseteq
    \clos{f^{-1}(L\comp)}\comp \cap \clos{D}\enspace,\\
    \intertext{or equivalently,}
    \clos{f^{-1}(L\comp)} \cup \clos{D}\comp \subseteq
    \ext{f}^{-1}(\clos{L}\comp) \cup \clos{D}\comp
    \enspace.
  \end{align*}

  The inclusion to be shown is thus implied by
  $\clos{f^{-1}(L\comp)} \subseteq \ext{f}^{-1}(\clos{L}\comp)$, that is, since
  $L$ is regular, by
  $\clos{f^{-1}(L\comp)} \subseteq \ext{f}^{-1}(\clos{L\comp})$.  As in the
  proof of the converse inclusion, the right-hand side being closed, this
  inclusion holds.
\end{proof}

\begin{lem}[Preservation Lemma]
  Let $f\colon A^* \to B^*$ be a Reg_continuous function and $E$ a set of
  equations that defines $\cV(A^*)$.  The function $f$ is $\cV$_continuous iff
  for all $(u = v) \in E$ and words $s, t \in A^*$,
  $\ext{f}(s\cdot u \cdot t) =_\cV \ext{f}(s \cdot v \cdot t)$.
\end{lem}
\begin{proof}
  \proofstep{Only if}%
  Suppose $f$ is $\cV$_continuous.  Let $u, v \in \pro{A^*}$ such that
  $u =_\cV v$, and $s, t \in A^*$.  Since by $\cV$_continuity
  $f^{-1}(B^*) \in \cV(A^*)$, either both $s\cdot u\cdot t$ and
  $s \cdot v\cdot t$ belong to the closure of this language, or they both do
  not.  The latter case readily yields the result, hence suppose we are in the
  former case.

  By definition, $u = \lim u_n$ and $v = \lim v_n$ for some Cauchy sequences
  of words ${(u_n)}_{n> 0}$ and ${(v_n)}_{n>0}$.  Since
  $s \cdot u \cdot t =_\cV s \cdot v \cdot t$, the hypothesis yields that
  $d_\cV(s\cdot u_n\cdot t, s\cdot v_n\cdot t)$ tends to $0$.  By
  Theorem~\ref{thm:ucont}, $f$ is uniformly continuous for $d_\cV$, hence
  $d_\cV(f(s\cdot u_n\cdot t), f(s\cdot v_n\cdot t))$ also tends to $0$ (note
  that both $f(s\cdot u_n\cdot t)$ and $f(s\cdot v_n\cdot t)$ are defined for
  all $n$ big enough).  This shows that
  $\ext{f}(s\cdot u\cdot t) =_\cV \ext{f}(s\cdot v\cdot t)$.

  \proofstep{If}%
  Suppose that $f$ preserves the equations of $E$ as in the statement.  Let
  $L \in \cV(B^*)$, we wish to verify that $L' = f^{-1}(L) \in \cV(A^*)$, or
  equivalently by definition, that $L'$ satisfies all the equations of $E$.  Let
  $(u = v) \in E$ be one such equation, and $s, t \in A^*$; we must show that
  $s\cdot u\cdot t \in \clos{L'} \Leftrightarrow s \cdot v \cdot t \in
  \clos{L'}$.


  By Lemma~\ref{lem:pinext}, since $f$ is Reg_continuous,
  \(\ext{f}(\clos{L'}) = \ext{f}(\ext{f}^{-1}(\clos{L})) \subseteq \clos{L}\).
  Now let $s \cdot u \cdot t \in \clos{L'}$, we thus have that
  $\ext{f}(s\cdot u \cdot t) \in \clos{L}$ (observe that
  $\ext{f}(s\cdot u\cdot t)$ is indeed defined).  By hypothesis,
  $\ext{f}(s\cdot u \cdot t) =_\cV \ext{f}(s\cdot v \cdot t)$; now since
  $L \in \cV(B^*)$, it must hold that $\ext{f}(s\cdot v \cdot t) \in \clos{L}$.
  Taking the inverse image of $\ext{f}$ on both sides, it thus holds that
  $s \cdot v \cdot t \in \ext{f}^{-1}(\clos{L})$, and Lemma~\ref{lem:pinext}
  then shows that $s \cdot v \cdot t \in \clos{L'}$.  As the argument works both
  ways, this shows that
  $s\cdot u\cdot t \in \clos{L'} \Leftrightarrow s \cdot v \cdot t \in
  \clos{L'}$, concluding the proof.
\end{proof}

Continuity can be seen as preserving \emph{membership} in $\cV$ (by inverse
image); this is where the nomenclature ``$\cV$_preserving function''
of~\cite{pin-silva17} stems from.  Strikingly, this could also be worded as
preserving \emph{nonmembership} in~$\cV$:

\begin{prop}
  A Reg_continuous total$\,$\footnote{In all the varieties we are interested in,
    one can easily modify any partial function into a total function while
    preserving its continuity properties.} function $f\colon A^* \to B^*$ is
  $\cV$_continuous iff for all $L \subseteq A^*$ that do \emph{not} belong to
  $\cV(A^*)$, $f(L)$ and $f(L\comp)$ are not $\cV$_separable.
\end{prop}
\begin{proof}
  We rely on a characterization due to Almeida~\cite[Lemma~3.2]{almeida99}:
  two languages $K$ and $L$ are $\cV$_separable iff for
  all $u \in \clos{K}, v \in \clos{L}$, we have that $u \neq_\cV v$.

  \proofstep{Only if}%
  Suppose $f$ is $\cV$_continuous, and let $L \subseteq A^*$ be a language
  outside $\cV(A^*)$.  There must be two profinite words $u, v \in \pro{A^*}$
  such that $u =_\cV v$, $u \in \clos{L}$ and $v \in \clos{L}\comp$.  By
  $\cV$_continuity and the Preservation Lemma, $\ext{f}(u) =_\cV \ext{f}(v)$,
  and moreover, $\ext{f}(u) \in \clos{f(L)}$ and
  $\ext{f}(v) \in \clos{f(L\comp)}$.  The characterization above thus implies
  that $f(L)$ and $f(L\comp)$ are not $\cV$_separable.

  \proofstep{If}%
  Assume that for all \(L \subseteq A^*\), if \(f(L)\) and \(f(L\comp)\) are
  \(\cV\)_separable, then \(L \in \cV(A^*)\).  For all $K \in \cV(B^*)$, we show
  that $L = f^{-1}(K) \in \cV(A^*)$.  Now \(f(L) \subseteq K\), and
  \(f(L\comp) = f(f^{-1}(K)\comp) = f(f^{-1}(K\comp)) \subseteq K\comp\).  Since
  \(K \in \cV(B^*)\), it is \(\cV\)_separable from its complement, hence
  \(f(L)\) and \(f(L\comp)\) are \(\cV\)_separable, and our assumption implies
  that \(L \in \cV(A^*)\).
\end{proof}


\subsection{The profinite extension of rational functions}

The Preservation Lemma already hints at our intention to see transducers as
computing functions from and to the free profinite monoids.  Naturally, if
$\tau$ is a rational function, its being Reg_continuous allows us to do so (by
Theorem~\ref{thm:ucont}).  For $u = \lim u_n$ a profinite word, we will write
$\tau(u)$ for $\ext{\tau}(u)$, i.e., the limit $\lim \tau(u_n)$, which exists by
continuity.  In this section, we develop a slightly more combinatorial approach
to the evaluation of \(\ext{\tau}\), and address two classes of profinite words:
those expressed as $s\cdot u\cdot t$ for $s, t$ words and $u$ a profinite word,
and those expressed as $x^\omega$ for $x$ a word.

Let \(\tau\) be an unambiguous transducer.  Recall that for any state $q$ of
\(\tau\) and any profinite word $u$, $q.u$ is well defined.  As a consequence, if
$s$ and $t$ are words, then there is at most one initial state~$q_0$, one
$q \in q_0.s$ and one $q' \in q.u$ such that $q'.t$ is final, and these states
exist iff $\tau(s\cdot u\cdot f)$ is defined.  Thus:

\begin{lem}\label{lem:eval}
  Let $\tau$ be an unambiguous transducer from $A^*$ to $B^*$, $s, t \in A^*$
  and $u \in \pro{A^*}$.  Suppose $\tau(s \cdot u \cdot t)$ is defined, and let
  $q_0, q, q'$ be the unique states such that $q_0$ is initial, $q \in q_0.s$,
  $q' \in q.u$, and $q'.t$ is final.  The following holds:
  \[\tau(s \cdot u \cdot t) = \Tif{\tau, \nop, q}(s) \cdot \Tif{\tau, q, q'}(u)
  \cdot \Tif{\tau, q', \nop}(t)\enspace.\]
\end{lem}
Let us now turn to the evaluation of $\omega$-terms:
\begin{lem}\label{lem:omega}
  Let $\tau$ be an unambiguous transducer from $A^*$ to $B^*$ and $x \in A^*$.
  If $\tau(x^\omega)$ is defined, then there are words $s, y, t \in B^*$ such
  that: \[\tau(x^\omega) = s\cdot y^{\omega - 1} \cdot t\enspace.\]
\end{lem}
\begin{proof}
  Consider a large value $n$; we study the behavior of $x^{n!}$ on $\tau$.
  There is an initial state $q_0$, a state $q$, and a final state $q_1$ such
  that $x^{n!}$ is accepted by a path going from $q_0$ to $q$ reading $x^i$,
  from $q$ to $q$ reading $x^k$ with $k < n$, and from $q$ to $q_1$ reading
  $x^j$.  Thus the accepting path for any word of the form $x^{m!}, m > n$ is
  similar to the one for $x^{n!}$: from $q_0$ to $q$, looping $(m! - n!)/k+1$
  times on $q$, and then from $q$ to $q_1$.  Let thus
  $s = \Tif{\tau, q_0, q}(x^i)$, $z = \Tif{\tau, q, q}(x^k)$, and
  $t = \Tif{\tau, q, q_1}(x^j)$.  It then holds that
  $\tau({(x^k)}^{m!}) = s\cdot z^{m! - (n!/k) + 1} \cdot t$.  Letting
  $c = n!/k - 1$, this shows that
  $\tau({(x^k)}^\omega) = s \cdot z^{\omega - c} \cdot t$.  Now on the one hand,
  ${(x^k)}^\omega = x^\omega$, and on the other hand, we similarly have that
  $z^{\omega - c} = y^{\omega - 1}$ by letting $y = z^c$.  We thus obtain
  that $\tau(x^\omega) = s \cdot y^{\omega - 1} \cdot t$.%
\end{proof}

These constitute our main ways to effectively evaluate the image of profinite
words through transducers.  Since they are ubiquitous in our study, we will
frequently apply these lemmas without explicitly citing them.

\subsection{The Syncing Lemma: Preservation Lemma applied to transducers}

We apply the Preservation Lemma on transducers and deduce a slightly more
combinatorial characterization of transducers describing continuous functions.
This does not provide an immediate decidable criterion, but our decidability
results will often rely on it.  The goal of the forthcoming lemma is to
decouple, when evaluating $s \cdot u \cdot t$ (with the notations of the
Preservation Lemma), the behavior of the $u$ part and that of the $s, t$ part.
This latter part will be tested against an \emph{equalizer} set:

\begin{defi}[Equalizer set]
  Let $u, v \in \pro{A^*}$.  The |equalizer set| of $u$ and $v$ in $\cV$ is:
  \[\equ_\cV(u, v) = \{ (s, s', t, t') \in {(A^*)}^4 \mid s\cdot u \cdot t
    =_\cV s' \cdot v \cdot t'\}\enspace.\]
\end{defi}

\begin{rem}\label{rk:almeida}
  The complexity of equalizer sets can be surprisingly high.  For instance,
  letting $\cV$ be the class of languages defined by
  $\{x^2 = x^3 \mid x \in A^*\}$, there is a profinite word $u$ for which
  $\equ_\cV(u, u)$ is undecidable (this relies on the existence of arbitrarily
  long square-free words).  On the other hand, equalizer sets quickly become
  less complex for common varieties; for instance, Lemma~\ref{lem:apeq} will
  provide a simple form for the equalizer sets of aperiodic supercancellative
  varieties.  \fullstop
\end{rem}

\begin{defi}[Input synchronization]
  Let $R, S \subseteq A^* \times B^*$.  The |input synchronization| of $R$ and
  $S$ is defined as the relation over $B^* \times B^*$ obtained by
  synchronizing the first component of $R$ and $S$:
  \[R \sync S = \{(u, v) \mid (\exists s)[ (s, u) \in R \land (s, v) \in
  S]\} \;\big(= S \circ R^{-1})\enspace.\]
\end{defi}

Naturally, the input synchronization of two rational functions is a rational
relation.

\begin{lem}[Syncing Lemma]
  Let $\tau$ be an unambiguous transducer from $A^*$ to $B^*$ and~$E$ a set
  of equations that defines $\cV(A^*)$.  The function $\tau$ is
  $\cV$_continuous iff:
  \begin{enumerate}
  \item $\tau^{-1}(B^*) \in \cV(A^*)$, and
  \item For any $(u = v) \in E$, any states $p, q$, any $p' \in p.u$, and any
    $q' \in q.v$, and letting $u' = \Tif{\tau,p,p'}(u)$ and
    $v' = \Tif{\tau, q, q'}(v)$:
    \[(\Tif{\tau, \nop, p} \sync \Tif{\tau, \nop, q}) \times (\Tif{\tau, p', \nop} \sync
    \Tif{\tau, q', \nop}) \subseteq \equ_\cV(u', v')\enspace.\]
  \end{enumerate}
\end{lem}
\begin{proof}
  We rely on the Preservation Lemma, since $\tau$ is Reg_continuous.

  \proofstep{Only if}%
  Suppose that $\tau$ is $\cV$_continuous, the first point is immediate.  For
  the second, we use the notation of the statement.  Let
  $(s, s', t, t') \in (\Tif{\tau, \nop, p} \sync \Tif{\tau, \nop, q}) \times
  (\Tif{\tau, p', \nop} \sync \Tif{\tau, q', \nop})$.  This implies that there
  are words $x, y \in A^*$ such that:
  \begin{itemize}
  \item $s = \Tif{\tau, \nop, p}(x), s' = \Tif{\tau, \nop, q}(x)$;
  \item $t = \Tif{\tau, p', \nop}(y), t' = \Tif{\tau, q', \nop}(y)$.
  \end{itemize}
  By Lemma~\ref{lem:eval}, we have that $\tau(x \cdot u
  \cdot y) = s \cdot u' \cdot t$ and $\tau(x \cdot v \cdot y) = s' \cdot v'
  \cdot t'$.  The Preservation Lemma then asserts that $s \cdot u' \cdot t =_\cV
  s' \cdot v' \cdot t'$, showing that $(s, s', t, t') \in \equ_\cV(u', v')$.

  \proofstep{If}%
  Let $(u = v) \in E$ and $x, y \in A^*$.  We must show that
  $\tau(x\cdot u \cdot y) =_\cV \tau(x \cdot v \cdot y)$.  Since
  $\tau^{-1}(B^*) \in \cV(A^*)$, either $\tau$ is defined on both
  $x\cdot u\cdot y$ and $x \cdot v \cdot y$, or on neither; in this latter case,
  the equality is satisfied by definition.  We thus suppose that both values are
  defined.  This implies that there are states $p,q,p',q'$ as in the statement,
  and using the same notation, letting $s, s', t, t'$ just as above, the
  hypothesis yields that $s \cdot u' \cdot t =_\cV s' \cdot v' \cdot t'$,
  showing the claim.
\end{proof}

\subsection{A profinite toolbox for the aperiodic setting}

In this section, we provide a few lemmas pertaining to our study of aperiodic
continuity.
We show that the equalizer sets of aperiodic supercancellative varieties are
well behaved.  Intuitively, the larger the varieties are, the more their
nonempty equalizer sets will be similar to the identity.  For instance, if
$s \cdot x^\omega =_\vA x^\omega$, for words $s$ and $x$, it should hold that
$s$ and $x$ have the same primitive root.  We first note the following easy fact
that will only be used in this section; it is reminiscent of the notion of
\emph{equidivisibility}, studied in the profinite context by
Almeida~and~Costa~\cite{almeida-costa17}.
\begin{lem}\label{lem:dance}
  Let $u, v$ be profinite words over an alphabet $A$ and $\cV$ be a
  supercancellative variety.  Suppose that there are $s, t \in A^*$ such that
  $u \cdot t =_\cV s \cdot v$, then there is a $w \in \pro{A^*}$ such that
  $u =_\cV s \cdot w$ and $v =_\cV w \cdot t$.  If moreover $u = v$ and $\cV$ is
  aperiodic, then $u =_\cV s \cdot u \cdot t$.
\end{lem}
\begin{proof}
  Let $u = \lim u_n$; if \({(u_n)}_{n>0}\) is ultimately constant, then this is
  immediate, so we assume that \(|u_n|\) is unbounded.  From
  $u \cdot t =_\cV s \cdot v$, and the fact that $s\cdot A^* \in \cV(A^*)$ by
  supercancellativity, we obtain that for $n$ large enough,
  $u_n \cdot t \in s \cdot A^*$.  Since $u$ is nonfinite, $|u_n| > |s|$ for $n$
  large enough, in which case $u_n = s \cdot w_n$ for some sequence
  ${(w_n)}_{n > 0}$.  Let $w \in \pro{A^*}$ be a limit point of this sequence,
  that exists by compactness (this is an essential property of the free profinite
  monoid, see, e.g.,~\cite[Theorem VI.2.5]{pin}).  It holds that
  $u = s \cdot w$.  Replacing $u$ by this value in the equation of the
  hypothesis, we thus have that $s \cdot w \cdot t =_\cV s \cdot v$, and since
  $\cV$ is supercancellative, that $v =_\cV w \cdot t$.

  For the last point, with $u = v$, we iterate the previous construction on $w$,
  since in that case, $u =_\cV w \cdot t =_\cV s \cdot w$.  This provides a
  sequence $w = w_1, w_2, w_3, \ldots$ such that
  $u =_\cV s^n\cdot w_n =_\cV w_n \cdot t^n$.  Taking a limit point $x$ of
  ${(w_n)}_{n>0}$, it thus holds that
  $u =_\cV s^\omega \cdot x =_\cV x \cdot t^\omega$, showing, by aperiodicity,
  that $u =_\cV s \cdot u =_\cV u \cdot t$.
\end{proof}
\begin{lem}\label{lem:apeq}
  Let $u, v$ be profinite words over an alphabet $A$ and $\cV$ be an aperiodic
  supercancellative variety.  Suppose $\equ_\cV(u, v)$ is nonempty.  There are
  words $x, y \in A^*$ and two pairs $\rho_1, \rho_2 \in {(A^*)}^2$ such that:
  \[\equ_\cV(u, v) = \Big(\Id\cdot \big((x^*, x^*)\rho_1^{-1}\big)\Big)
  \times \Big(\big(\rho_2^{-1}(y^*, y^*)\big)\cdot
  \Id\Big)\enspace.\]
\end{lem}
\begin{proof}
  Let us first establish the property for $u = v$.  Assume that there are
  \emph{nonempty primitive} words $x, y$ such that $x \cdot u \cdot y =_\cV u$;
  we show the statement of the lemma with these $x$ and $y$, and
  $\rho_1 = \rho_2 = (1, 1)$.  Note that
  $x^\omega \cdot u \cdot y^\omega =_\cV u$, hence, since
  $x^{\omega+1} = x^\omega$ and similarly for $y$, we have that
  $x \cdot u =_\cV u \cdot y =_\cV u$.  This and the fact that $\cV$ is
  supercancellative show the right-to-left inclusion.

  For the left-to-right inclusion, let $s, s', t, t'$ be such that
  $s\cdot u \cdot t =_\cV s' \cdot u \cdot t'$.  Since $\cV$ is supercancellative,
  this implies that the equation also holds if common prefixes of $s$ and $s'$
  and common suffixes of $t$ and $t'$ are removed.  We may thus assume that we
  are in two possible situations, by symmetry:
  \begin{enumerate}
  \item Suppose $s' = t' = 1$, that is, $s \cdot u \cdot t =_\cV u$, and that
    $s, t$ are nonempty.  By the same token as above, this shows that
    $s \cdot u =_\cV u \cdot t =_\cV u$.  In particular, this implies that:
    \[s^{|x|} \cdot u \cdot t^{|y|} =_\cV x^{|s|} \cdot u \cdot
      y^{|t|}\enspace,\]
    which implies, since $\cV$ is supercancellative, that $s^{|x|} = x^{|s|}$ and
    $t^{|y|} = y^{|t|}$.  As $x$ and $y$ are primitive, this shows that $s \in
    x^*$ and $t \in y^*$.  (Note that this holds even if one of $s$ or $t$ is
    empty.)
  \item Suppose $s = t' = 1$, that is, $u \cdot t =_\cV s' \cdot u$, and that
    $s', t$ are nonempty.  By
    Lemma~\ref{lem:dance}, we have that $u =_\cV s' \cdot u \cdot t$, and we can
    appeal to the previous situation, showing that $s' \in x^*$ and $t \in y^*$.
  \end{enumerate}
  (The cases where one of $s, t$ is empty, in the first point, or one of $s', t$
  is empty, in the second, are treated similarly.  Note that it is not possible
  for both $s$ and $s'$ to be nonempty, since that would imply that they start
  with different letters, falsifying the assumed equation by
  supercancellativity.)

  We assumed that the $x, y$ existed, we ought to show the other cases satisfy
  the claim.  The two situations above show that if $\equ_\cV(u,u)$ is nonempty,
  then such $x, y$ exist, although without the guarantee that they be nonempty.
  Now if $x \cdot u =_\cV u$ and there are no nonempty $y$ such that
  $x \cdot u \cdot y =_\cV u$, this implies that there are no nonempty $y$ such
  that $u \cdot y =_\cV u$.  Consequently, in the above case, $t = t' = 1$,
  and the analysis stands.  This concludes the proof for the case $u = v$.

  We will reduce the case $u \neq v$ to this one.  Indeed, suppose that
  $s \cdot u \cdot v = s' \cdot v \cdot t'$.  Again, by stripping away common
  prefixes and suffixes, we are faced with two cases:
  \begin{enumerate}
  \item Suppose $s' = t' = 1$, that is, $s \cdot u \cdot t =_\cV v$.  We have
    that $\equ_\cV(u, v) = \equ_\cV(u, s \cdot u \cdot t)$, hence
    $(m, m', n, n') \in \equ_\cV(u, v)$ iff
    $(m, m'\cdot s, n, t \cdot n') \in \equ_\cV(u, u)$, and the result follows.
  \item Suppose $s = t' = 1$, that is, $u \cdot t =_\cV s' \cdot v$.  By
    Lemma~\ref{lem:dance}, there is a profinite word $w$ such that $u =_\cV s'
    \cdot w$ and $v =_\cV w \cdot t$, hence $(m, m', n, n') \in \equ_\cV(u, v)$
    iff $(m \cdot s', m', n, t \cdot n') \in \equ_\cV(w, w)$, concluding the
    proof.\qedhere
  \end{enumerate}
\end{proof}


\begin{lem}\label{lem:apeqomega}
  Let $x, y$ be words.  For every aperiodic supercancellative variety $\cV$, the
  equality $\equ_\cV(x^\omega, y^\omega) = \equ_\vA(x^\omega, y^\omega)$ holds.
\end{lem}
\begin{proof}
  The inclusion from right to left is clear, since all equations true in \vA
  hold in \cV.

  \def\subp{_\text{\normalfont p}}\def\subs{_\text{\normalfont s}}
  In the other direction, let us write $u = x^{|y|}$ and $v = y^{|x|}$; we have
  that $u^\omega = x^\omega$ and $v^\omega = y^\omega$.  Suppose
  $s \cdot u^\omega t =_\cV s' \cdot v^\omega \cdot t'$.  In particular, since
  $\cV$ is supercancellative, this means that $s \cdot u^n$ is a prefix of
  $s' \cdot v^n$, or vice-versa, depending on whether $|s| > |s'|$ or the
  opposite.  This implies that $u \cdot u\subp = v\subs \cdot v$ for some prefix
  $u\subp$ of $u$ and suffix $v\subs$ of $v$.  Hence (by,
  e.g.,~\cite[Proposition~1.3.4]{lothaire97}) $u$ and $v$ are conjugate.  Their
  respective primitive roots are thus conjugate
  (by~\cite[Proposition~1.3.3]{lothaire97}); writing  $z\cdot z'$ and
  $z' \cdot z$ for them, we have that $u^\omega = {(z\cdot z')}^\omega$ and
  $v^\omega = {(z'\cdot z)}^\omega$.

  Thus the equation above reads:
  $s \cdot {(z \cdot z')}^\omega \cdot t =_\cV s' \cdot {(z' \cdot z)}^\omega \cdot
  t'$.  As in the proof of Lemma~\ref{lem:apeq}, removing the common prefixes and
  suffixes (which we can do both in $\cV$ and $\vA$), we are left with two
  possibilities:
  \begin{itemize}
  \item Suppose $s' = t' = 1$, that is, $s \cdot {(z \cdot z')}^\omega \cdot t =_\cV
    {(z' \cdot z)}^\omega$.  The same argument as in Lemma~\ref{lem:apeq} shows that
    $s \in z' \cdot {(z\cdot z')}^*$ and $t \in {(z \cdot z')}^*\cdot z$, and hence
    the equation holds in $\vA$ too;
  \item Suppose $s = t' = 1$, that is, ${(z \cdot z')}^\omega \cdot t =_\cV s'
    \cdot {(z' \cdot z)}^\omega$.  Similarly, as $\cV$ is supercancellative and
    aperiodic, this shows that $s' \in z \cdot {(z' \cdot z)}^*$ and $t \in {(z
    \cdot z')}^*\cdot z$, and the equation holds in $\vA$ too, concluding the
    proof.\qedhere
  \end{itemize}
\end{proof}

\begin{rem}
  For two aperiodic supercancellative varieties $\cV$ and $\cW$, we could
  further show that if both $\equ_\cV(u, v)$ and $\equ_\cW(u, v)$ are nonempty,
  then they are equal, for any profinite words $u, v$.  It may however happen
  that one equalizer set is empty while the other is not; for instance, with
  $u = {(ab)}^\omega$ and $v = {(ab)}^\omega\cdot a \cdot {(ab)}^\omega$, the
  equalizer set of $u$ and $v$ in $\vDA$ is nonempty, while it is empty in
  $\vA$.
  \fullstop
\end{rem}

\section{Intermezzos}\label{sec:intermezzos}

We present a few facts of independent interest on continuous rational functions.
Through this, we develop a few examples, showing in particular how the
Preservation and Syncing Lemmas can be used to show (non)continuity.  In a first 
part, we study when the structure of the transducer is relevant to continuity,
and in a second, when the (non)inclusion of variety relates to (non)inclusion of 
the class of continuous rational functions.

\subsection{Transducer structure and continuity}\label{sec:struct}




As noted by \reutschu~\cite[p.~231]{reutenaeur-schutzenberger95}, there exist
numerous natural varieties~$\cV$ for which any $\cV$_realizable rational
function is $\cV$_continuous.  Indeed:
\begin{prop}\label{prop:transtocont}
  Let $\cV$ be a variety of languages closed under inverse $\cV$_realizable
  rational function.  Any $\cV$_realizable rational function is
  $\cV$_continuous.  This holds in particular for the varieties $\vA$, $\vGsol$, and
  $\vG$.
\end{prop}
\begin{proof}
  This is due to a classical result of Sakarovitch~\cite{sakarovitch79} (see
  also~\cite{pin-sakarovitch85}), stating, in modern parlance, that a variety
  $\cV$ is closed under \emph{block product} iff it is closed under inverse
  $\cV$_realizable rational functions (note that there has been some fluctuation
  on vocabulary, since wreath product was used at some point to mean block
  product).  That $\vA$, $\vGsol$, and $\vG$ are closed under block product is
  folklore.
\end{proof}

This naturally fails for all our other varieties, since they are not closed
under inverse $\cV$_realizable rational functions.  For completeness, we give
explicit constructions in the proof of the following Proposition.
\begin{prop}
  For $\cV \in \{\vJ, \vL, \vR, \vDA, \vAb, \vGnil, \vCom\}$, there are
  $\cV$_realizable rational functions that are not $\cV$_continuous.
\end{prop}
\begin{proof}
  We devise simple counter examples with $A = \{a, b\}$.

  \proofstep{The $\vJ, \vR$ and $\vCom$ cases}%
  Recall that $A^*a \notin \vR(A^*) \cup \vCom(A^*)$.  The minimal unambiguous
  two-state transducer $\tau$ that erases all of its input except for the last
  letter is a $(\vJ \cap \vCom)$_transducer; indeed, $a$ acts in the same way as
  $b$ and they are idempotent on the transducer.  However,
  $\tau^{-1}(a) = A^*a$.

  \proofstep{The $\vR$ and $\vDA$ cases}%
  Consider the \emph{Dyck language} $D$ over $A$; this is the (nonregular)
  language of well-parenthesized expressions where $a$ is the opening and $b$
  the closing parenthesis.  Write $D^{(k)}$ for the Dyck language where
  parentheses are nested at most $k$~times, for instance
  $D^{(0)} = 1, D^{(1)}={(ab)}^*$ and $D^{(2)} = {(a{(ab)}^*b)}^*$.  These languages
  have great importance in algebraic language theory, as they separate each
  level of the \emph{dot-depth hierarchy}~\cite{brzozowski-knast78}.  It holds
  in particular that $D^{(1)} \notin \vDA(A^*)$.

  Let $\tau$ be the rational function that removes the first letter of each
  block of $a$'s and each block of $b$'s; naturally, $\tau$ is $\vL$_realizable.
  However, $\tau^{-1}(D^{k-1}) = D^k$, showing not only that $\tau$ is not
  continuous for \vDA, but also not continuous for \emph{any} level of the
  dot-depth hierarchy.

  \proofstep{The $\vGnil$ and $\vAb$ cases}%
  Consider the two-state transducer $\tau$ where $a$ loops on both states, and a
  $b$ on one state goes to the other.  When $a$ is read on the first state, it
  produces a $x$, while all the other productions are the identity.  This is an
  $\vAb$_transducer.  However, $\tau(aba) = xba \neq baa = \tau(baa)$, hence it
  is not $\vAb$_continuous by the Preservation Lemma, since $aba =_\vAb baa$.
  For $\vGnil$, let $L$ be the language over $\{a, b, x\}$ with a number of $x$
  congruent to $0$ modulo $3$.  It can be shown that
  $\tau^{-1}(L) \notin \vGnil(A^*)$, intuitively since this language needs to
  differentiate between those $a$'s that are an even number of $b$'s away from
  the beginning of the word, and those which are not.
\end{proof}

The converse concern, that is, whether all $\cV$_continuous rational
functions are $\cV$_realizable, was mentioned by
\reutschu~\cite{reutenaeur-schutzenberger95} for $\cV = \vA$.
\begin{prop}\label{prop:conttrans}
  For $\cV \in \{\vJ, \vL, \vR, \vDA, \vA, \vAb, \vCom\}$, there are
  $\cV$_continuous rational functions that are not $\cV$_realizable.
\end{prop}
\begin{proof}
  \proofstep{The aperiodic cases}%
  Let $A = \{a\}$, a unary alphabet.  Consider the transducer $\tau$ that
  removes every second $a$: its minimal transducer not being a $\vA$_transducer,
  it is not $\vA$_realizable (this is a property of subsequential
  transducers~\cite{reutenaeur-schutzenberger95}).  However, all the unary
  languages of $\cV$ are either finite or co-finite, and hence for any
  $L \in \cV(A^*)$, $\tau^{-1}(L)$ is either finite or co-finite, hence belongs
  to $\cV(A^*)$.

  \proofstep{The $\vAb$ and $\vCom$ cases}%
  Over $A = \{a, b\}$, define $\tau$ to map words $w$ in $aA^*$ to ${(ab)}^{|w|}$,
  and words $w$ in $bA^*$ to ${(ba)}^{|w|}$.  Clearly, $a$ and $b$ cannot act
  commutatively on the transducer.  Now $\tau(ab) =_\vCom \tau(ba)$, and
  moreover $\tau(x^\omega) =_\vAb {(ab)}^\omega =_\vAb 1 = \tau(1)$, hence $\tau$
  is continuous for both $\vAb$ and $\vCom$ by the Preservation Lemma.
\end{proof}


We delay the positive answers to that question, namely for
$\vGnil, \vGsol, \vG$, to Corollary~\ref{cor:grpconttotrans} as they constitute
our main lever towards the decidability of continuity for these classes.

\subsection{Variety inclusion and inclusion of classes of continuous functions}\label{sec:inc}

In this section, we study the consequence of variety (non)inclusion on the 
inclusion of the related classes of continuous rational functions.  This is
reminiscent of the notion of \emph{heredity} studied by~\cite{pin-silva11},
where a function is $\cV$_hereditarily continuous if it is $\cW$_continuous for
each subvariety $\cW$ of $\cV$.  Variety noninclusion provides the simplest
study case here:


\begin{prop}
  Let $\cV$ and $\cW$ be two varieties.  If $\cV \not\subseteq \cW$ then
  there are $\cV$_continuous rational functions that are not
  $\cW$_continuous.
\end{prop}
\begin{proof}
  Let $L \in \cV(A^*)$ be such that $L \notin \cW(A^*)$.  Define
  $f\colon A^* \to A^*$ as the identity function with domain $L$.  Clearly,
  as $f^{-1}(K) = K \cap L$, the function $f$ is $\cV$_continuous.  However,
  $f^{-1}(A^*) = L \notin \cW(A^*)$ and $A^* \in \cW(A^*)$, thus $f$ is not
  $\cW$_continuous.
\end{proof}

The remainder of this section focuses on a dual statement:

\begin{center}
  \vspace{.5em}
  \emph{If $\cV \subsetneq \cW$, are all $\cV$_continuous rational
    functions $\cW$_continuous?}
\end{center}

\subsubsection{The group cases}

We first focus on group varieties.  Naturally, if 1. $\cV$_continuous rational
functions are $\cV$_realizable and 2. $\cW$_realizable rational functions are
$\cW$_continuous, this holds.  Appealing to the forthcoming
Corollary~\ref{cor:grpconttotrans} for point 1 and
Proposition~\ref{prop:transtocont} for point 2, we then get:
\begin{prop}
  For $\cV, \cW \in \{\vGnil, \vGsol, \vG\}$ with $\cV \subsetneq \cW$, all
  $\cV$_continuous rational functions are $\cW$_continuous.  This however fails
  for $\cV = \vAb$ and for any $\cW \in \{\vGnil, \vGsol, \vG\}$.
\end{prop}
\begin{proof}
  It remains to show the case $\cV = \vAb$.  This is in fact the same example as
  in the proof of Proposition~\ref{prop:conttrans}, to wit, over $A = \{a, b\}$,
  the rational function $\tau$ that maps $w \in aA^*$ to ${(ab)}^{|w|}$, and words
  $w \in bA^*$ to ${(ba)}^{|w|}$.  Indeed, we saw that this function is continuous
  for $\vAb$, but we have that $\tau(a) = ab$ on the one hand, and
  $\tau(b^\omega a) = {(ba)}^\omega ba =_\cW ba$, but $ab \neq_\cW ba$.  The
  Preservation Lemma then shows that $\tau$ is not continuous for $\cW$.
\end{proof}

\begin{prop}\label{prop:abtocom}
  All $\vAb$_continuous rational functions are $\vCom$_continuous.
\end{prop}
\begin{proof}
  Indeed, if $u =_\vAb v$ with $u, v$ words, then $u =_\vCom v$, since these
  varieties separate the same words.  As $\vCom$ is defined using equations on
  words, this directly shows the claim by the Preservation Lemma.
\end{proof}

\subsubsection{The aperiodic cases}

We now turn to aperiodic varieties.  For less expressive varieties, the
property fails:
\begin{prop}\label{prop:apnon}
  For $\cV \in \{\vJ, \vL, \vR\}$ and $\cW \in \{\vL, \vR, \vDA, \vA\}$ with
  $\cV \subsetneq \cW$, there are $\cV$_continuous rational functions that
  are not $\cW$_continuous.
\end{prop}
\begin{proof}
  Define $\tau\colon {\{a\}}^* \to {\{a, b\}}^*$ to be the rational function that
  changes every other $a$ to $b$; that is, $\tau(a^{2n}) = {(ab)}^n$, and
  $\tau(a^{2n+1})={(ab)}^n\cdot a$.  Note that naturally, over a single letter,
  $a\cdot {(aa)}^\omega = {(aa)}^\omega \cdot a = a^{\omega+1}$.  Now
  $\tau(a^\omega) = {(ab)}^\omega$ and $\tau(a^{\omega+1})= {(ab)}^\omega\cdot a$,
  and since these two profinite words are equal in $\vJ$ and $\vR$, the
  Preservation Lemma shows that $\tau$ is continuous for both $\vJ$ and $\vR$.
  However, these two profinite words are not equal in $\vL$, $\vDA$, and $\vA$,
  showing that $\tau$ is continuous for none of those varieties.

  The remaining case, that is, showing the existence of a $\vJ$_continuous
  rational function that is not $\vR$_continuous is done symmetrically, with the
  function mapping $a^{2n}$ to ${(ab)}^n$ and $a^{2n+1}$ to $b\cdot {(ab)}^n$.
\end{proof}

\begin{prop}
  Any $\vDA$_continuous rational function is $\vA$_continuous.
\end{prop}
\begin{proof}
  First note that both $\vDA$ and $\vA$ satisfy the hypotheses of
  Lemma~\ref{lem:apeq}.  Consider a $\vDA$_continuous rational function
  $\tau\colon A^* \to B^*$.  By the Syncing Lemma, to show that it is
  $\vA$_continuous, it is enough to show that 1.~$\tau^{-1}(B^*) \in \vA(A^*)$,
  and 2.~That some input synchronizations of $\tau$, based on equations of the
  form $x^\omega =_\vA x^{\omega + 1}$, belong to an equalizer set of the form
  (by Lemma~\ref{lem:omega}):
  \[
    \equ_\vA(\alpha \cdot y^\omega \cdot \beta,\; \alpha' \cdot z^\omega \cdot
    \beta')  = \{(s, s', t, t') \mid (s\cdot \alpha,\; s' \cdot \alpha', \;\beta\cdot
    t,\; \beta'\cdot t') \in \equ_\vA(y^\omega,\; z^\omega)\}\enspace.
  \]

  Applying the Syncing Lemma on $\tau$ for the variety $\vDA$, we get that
  point~1 is true, since $\tau^{-1}(B^*) \in \vDA(A^*)$.  Similarly, point~2 is
  true since $x^\omega = x^{\omega+1}$ is an equation of $\vDA$, and
  Lemma~\ref{lem:apeqomega} implies that the equalizer set of
  the equation above is the same in $\vDA$ and $\vA$.
\end{proof}

However, this property does not hold beyond \emph{rational} functions:
\begin{prop}
  There are nonrational functions that are continuous for both $\vDA$ and Reg
  but are not $\vA$_continuous.
\end{prop}
\begin{proof}
  Define $f\colon {\{a\}}^* \to {\{a, b\}}^*$ by $f(a^{2n}) = {(ab)}^n$ and
  $f(a^{2n+1}) = {(ab)}^n\cdot a \cdot {(ab)}^n$.  We first have to check that $f$
  is indeed Reg_continuous.  Given a regular language $L$, we define a
  pushdown automaton over ${\{a\}}^*$ that recognizes $f^{-1}(L)$; since all
  unary context-free languages are regular, by Parikh's theorem, this shows the
  claim.  If the input is of the form $a^{2n}$, then the pushdown automaton may
  check that ${(ab)}^n \in L$, by simulating the automaton for $L$.  If the input
  is of the form $a^{2n+1}$, then the pushdown automaton can guess the middle
  position of the input, and accordingly check that ${(ab)}^n \cdot a \cdot {(ab)}^n
  \in L$, again using the automaton for $L$.  This concludes the construction.

  The function $f$ being Reg_continuous, consider its extension $\ext{f}$.  As
  in Proposition~\ref{prop:apnon}, checking that $f$ is $\vDA$_continuous
  amounts to checking that $\ext{f}(a^{\omega}) =_\vDA \ext{f}(a^{\omega+1})$.  The
  left-hand side being ${(ab)}^\omega$ while the right-hand side is
  ${(ab)}^\omega \cdot a \cdot {(ab)}^\omega$, this holds.  However, these two
  profinite words are equal in $\vDA$ but not in $\vA$, hence this function is
  not $\vA$_continuous, again appealing to the Preservation Lemma.
\end{proof}

\section{Deciding continuity for transducers}\label{sec:dec}
\subsection{Deciding continuity for group varieties}\label{sec:group}

\reutschu showed in~\cite{reutenaeur-schutzenberger95} that a rational function
is $\vG$_continuous iff it is $\vG$_realizable.  Since this is proven
effectively, it leads to the decidability of $\vG$_continuity.  In
Proposition~\ref{prop:transtocont}, we saw that the right-to-left statement also
holds for $\vGsol$; we now show that the left-to-right statement holds for all
group varieties $\cV$ that contain $\vGnil$---this odd-looking condition on
\(\cV\) is here to ensure that the \emph{free group} is embedded, in a precise
sense, in~$\cV$~\cite[\S~6.1.9]{robinson95}.  As
in~\cite{reutenaeur-schutzenberger95}, but with sensibly different techniques,
we show that $\cV$_continuous transducers are plurisubsequential.  The Syncing
Lemma will then imply that such transducers are $\cV$_transducers.  Both
properties rely on an omnibus normal form which uses the three following
notions:
\begin{defi}
  \begin{itemize}
  \item A transducer \(\tau\) is \emph{utilitarian} if for any two states \(p, q\):
  \[\Big[(\exists x, y) \big[\emptyset \;\neq\; (\Tif{\tau, p, \nop} \sync
  \Tif{\tau, q, \nop}) \;\subseteq\; (x, y) \cdot \Id\big]\Big] \Rightarrow p =
  q\enspace.\]%
  \item Within a transducer, a \emph{pan} is a triplet \((q, q', p)\) of the form:
  \begin{center}
    \scalebox{0.7}
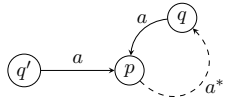
{\begin{tikzpicture}[>=stealth,bend angle=40]
  \useasboundingbox (-2, 1.4) rectangle (2, -0.5);
  \node (p) [state] {$p$};
  \node (qp) at ($(p)-(2,0)$) [state] {$q'$};
  \node (q) at ($(p)+(1,1)$) [state] {$q$};
  \draw[->,bend right] (q) edge node[pos=0.5,above] {$a$} (p);
  \draw[->] (qp) edge node[pos=0.5,above] {$a$} (p);
  \draw[->,dashed] (p) .. controls +(1,-1) and +(1,-1) .. (q)
  node[pos=0.5,yshift=0em,right] {$a^*$};
\end{tikzpicture}

}
  \end{center}
  This is regardless of the outputs.  The pan is \emph{proper} if \(q \neq q'\).
\item Within a transducer, and for a letter \(a\), a state \(q\) is
  \emph{\(a\)_recurrent} if \(p.a^\omega = \{p\}\).  The transducer itself is said to be
  \(a\)_recurrent if all of its states are, and recurrent if it is
  \(a\)_recurrent for all letters.
\end{itemize}
\end{defi}

\begin{lem}\label{lem:nf}
  Let $\tau$ be a transducer.  There is a utilitarian transducer
  \(\tau'\) computing the same function.  Additionally:
  \begin{itemize}
  \item If \(\tau\) is \(\cV\)_continuous for some group variety that contains
    \(\vGnil\), then \(\tau'\) has no proper pan;
  \item For any letter \(a\), if \(\tau\) is \(a\)_recurrent, so is \(\tau'\).
  \end{itemize}
\end{lem}
\begin{proofof}{Lemma~\ref{lem:nf}}
  We start with the utilitarianism, then focus on the first of the
  ``additionally'' parts.  The second point therein is easily satisfied by
  construction.  Utilitarianism is shown in three steps.

  \proofstep{Step 1}%
  Define $\tau'$ to be the Cartesian product of $\tau$ with the powerset automaton
  of its reversal.  By construction, we get the following fact, writing \(R_p\) for
  the language of words accepted from \(p\), that is \(\{w \mid p.w \cap F \neq \emptyset\}\):
  \begin{fact}\label{fact:rp}
    For any states $p, p'$ of $\tau'$, $R_p \cap R_{p'} \neq \emptyset \Rightarrow
    R_p = R_{p'}$.
  \end{fact}

  \proofstep{Step 2}%
  Write simply $\tau$ for the result of Step~1.  In this step, we make sure that
  outputs are produced ``as soon as possible'', a process known as
  \emph{normalization} (e.g.,~\cite[Section~1.5.2]{lothaire02}) that we sketch
  for completeness.  For every state $q$, write $\pi_q$ for the longest string
  such that $\Tif{\tau, q, \nop}(A^*) \subseteq \pi_q\cdot B^*$.  Now define the
  new output function $(\lambda', \mu', \rho')$ by letting:
  \[\mu'(1, q) = \mu(1, q) \cdot \pi_q,\quad
  \mu'(q, a, q') = \pi_q^{-1}\mu(q, a, q')\cdot \pi_{q'},\quad \mu'(q, 1) =
  \pi_q^{-1}\mu(q, 1)\enspace.\]%
  Write $\tau'$ for $\tau$ equipped with the output function $\mu'$.  For no
  state $q$ there is a letter $b \in B^*$ such that
  $\Tif{\tau, q, \nop}(A^*) \subseteq b\cdot B^*$.

  \proofstep{Step 3}%
  Write again $\tau$ for the result of the previous step.  Naturally, $\tau$ still
  satisfies Fact~\ref{fact:rp}.  Consider two states $p, q$ such that there are
  $x, y \in B^*$ satisfying
  $\emptyset \neq (\Tif{\tau, p, \nop} \sync \Tif{\tau, q, \nop}) \subseteq (x, y)\cdot \Id$.  The first part
  of this assumption implies that $R_p \cap R_q \neq \emptyset$, and thus, by
  Fact~\ref{fact:rp}, $R_p = R_q$.  In other words, $\Tif{\tau, p, \nop}$ and
  $\Tif{\tau, q, \nop}$ have the same domain.  The second part of the assumption
  thus indicates that every production of $p$ (\resp of $q$) starts with $x$
  (\resp with $y$), and Step~2 asserts that $x = y = 1$.  Hence
  $\Tif{\tau, p, \nop}$ and $\Tif{\tau, q, \nop}$ actually compute the same function.
  We can thus merge them into a single state without changing the function
  realized.  Repeating this operation results in a transducer \(\tau'\) that is
  utilitarian.

  \proofstep{No proper pans}%
  Consider a pan $(q, q', p)$ on $a$ in $\tau'$.  As $p$ can be reached from both
  $q$ and $q'$ reading $a$, the product
  $P = \Tif{\tau', q, \nop} \sync \Tif{\tau', q', \nop}$ is nonempty.  Write
  $x = \Tif{\tau', q', p}(a), y = \Tif{\tau', p, q}(a^n), z = \Tif{\tau', q, p}(a)$, for
  some $n$ such that $y$ is defined.  We let $h$ be the longest common suffix of
  $x$ and $z$, and $x = x'\cdot h$ and $z = z' \cdot h$.

  As $\tau$ is $\vG$_continuous, let us apply point 2 of the Syncing Lemma on
  $\tau'$, the equation $(a^\omega = 1)$, from the pair of states $(q', q')$ to
  $(q, q')$.  With $Z = \Tif{\tau, \nop, q'} \sync \Tif{\tau, \nop, q'}$, a nonempty
  subset of the identity, we have that
  $Z \times P \subseteq \equ_\cV(x{(yz)}^{\omega - 1}y, 1)$.  We write, in the following,
  $\nu^{-1}$ for $\nu^{\omega-1}$, to convey the fact that $\nu^{-1}$ is the inverse of
  $\nu$ in $\cV$; that is: $\nu\cdot \nu^{-1} =_\cV 1$ (this analogy naturally carries
  further, since for instance, ${(\nu\eta)}^{-1} =_\cV \eta^{-1}\cdot \nu^{-1}$).  Let
  $(s, s, u, u') \in Z \times P$, then:
  \begin{align*}
    s \cdot u' & =_\cV s \cdot x\cdot {(yz)}^{\omega - 1} \cdot y \cdot u \tag{By the
             Syncing Lemma}\\
           & =_\cV s \cdot x \cdot z^{-1}\cdot y^{-1} \cdot y \cdot u\\
           & =_\cV s \cdot x \cdot z^{-1} \cdot u\\
           & =_\cV s \cdot (x'h) \cdot {(z'h)}^{-1} \cdot u =_\cV s \cdot x'
             \cdot z'^{-1} \cdot u\enspace.
  \end{align*}
  By cancellation, this shows that $u' =_\cV x' \cdot z'^{-1} \cdot u$.  Since
  $u'$ is a word and $x'$ and $z'$ do not share a common suffix, there is a word
  $w$ such that $u = z'\cdot w$ (this is true in the \emph{free group}, which is
  embedded, in a precise sense, in~$\cV$~\cite[\S~6.1.9]{robinson95}).  This
  implies that $u' =_\cV x' \cdot w$ and shows that
  $P \subseteq (z', x')\cdot \Id$, hence that $q = q'$ by the main property of this lemma,
  and the pan is not proper.
\end{proofof}

\begin{lem}\label{lem:nftocV}
  Let $\cV$ be a variety of group languages that contains $\vGnil$.  For any
  $\cV$_continuous unambiguous transducer, there is an equivalent
  plurisubsequential $\cV$_transducer.
\end{lem}
\begin{proofnoqed}
  Let \(\tau\) be a \(\cV\)_continuous unambiguous transducer.  The proof is split in
  three facts:
  \begin{fact}\label{fact:rec}
    There is a utilitarian recurrent transducer \(\tau'\) that defines the same
    function as \(\tau\).
  \end{fact}

  \begin{fact}\label{fact:rectosub}
    Any recurrent transducer is plurisubsequential.
  \end{fact}

  \begin{fact}\label{fact:subtov}
    Any \(\cV\)_continuous utilitarian plurisubsequential transducer is a
    \(\cV\)_transducer.
  \end{fact}

  These facts together naturally imply the lemma.

  \begin{proofof}[Proof of Fact~\ref{fact:rec}]{Fact~\ref{fact:rec}}
    We first apply Lemma~\ref{lem:nf} on \(\tau\).  We turn the resulting
    transducer---which we call \(\tau\) again---into a recurrent transducer one letter at
    a time.  After each letter, we apply again Lemma~\ref{lem:nf}, thus
    obtaining at the end of the process a utilitarian and recurrent
    transducer equivalent to \(\tau\).

    In the following, for any state \(q\), we write $L_q$ for
    $\{w \mid q \in I.w\}$ and $R_q$ for $\{w\mid q.w \cap F \neq \emptyset\}$.  The preimage of
    $\tau$ is denoted by $L=\tau^{-1}(B^*)$. Since $\tau$ is $\cV$_continuous, then
    $L$ is in $\cV$.  Finally, we say that a state $p$ is weakly $a$_recurrent
    if $p\in p.a^w$.

    Let \(a\) be a letter.  We first perform the direct product of the automaton
    with an automaton remembering the last weakly $a$_recurrent state seen.
    Let $p$ be a non weakly $a$_recurrent state.  Now $\tau$ is $\cV$-continuous and has
    no proper pan, hence for any $u,v$ such that \(u\in L_p\) and
    \(v\in R_p\), there exist $q, q'$ such that:
    \def\tto#1{\stackrel{#1}{\longrightarrow}}%
    \[I \tto{u} q \tto{a^\omega} q \tto{a^n} q' \tto{v} F\enspace.\]%
    By the same argument used in showing the absence of proper pan in
    Lemma~\ref{lem:nf}, it holds that $p=q'$.  Since $p$ is connected to some
    weakly $a$_recurrent state, there is some $k > 0$ such that the above path
    can be decomposed as follows:
    \[I \tto{u} q \tto{a^{\omega}} q \tto{a^{\omega - k}} q'' \tto{a^k} p \tto{v}
    F\enspace.\]%
    Since $p$ contains the information of the last weakly $a$_recurrent state
    seen, the choice of $q''$ is independent from $u$ and $v$.  Hence the choice
    of $q$ too is independent from both $u$ and $v$.  In particular, this shows
    that $L_p\subseteq L_q$.  Furthermore, for any word in $u\in L_q$ and any word
    $v\in R_q$, we have $ua^\omega v \in \clos{L}$, and thus $uv \in L$ as well.
    Hence there exists a state $r$ such that:
    \[I \tto{u} r \tto{v} F\enspace.\]
    By the same argument as above, we necessary have $r=p$, proving that
    $L_p=L_q$.

    We now merge together all states $p$, $q$ satisfying the property that
    $L_p=L_q$.  Since, by Lemma~\ref{lem:nf}, Fact~\ref{fact:rp} we have either
    $R_p\cap R_q = \emptyset$ or $R_p = R_q$, merging $p$ and $q$ will not change the
    function computed.  After these merges, all states are $a$-recurrent.

    We now ensure that merging these states preserves that states were recurrent
    for other letters.  Assume that the transducer was \(b\)_recurrent.  Consider
    $p$ and $q$ to be merged; since $L_p=L_q$, they are both in a cycle of $b$'s
    of the exact same length, say $n$.  Let $p'\in p.b^k$ and $q'\in q.b^k$.  We
    finish this proof by showing that $L_{p'} = L_{q'}$, which implies that $p'$
    and $q'$ are also merged during the process.

    For any $u\in L_{p'}$, we have $ub^{n-k} \in L_{p} = L_{q}$. Hence, $ub^n$ is in
    $L_{q'}$.  Since the transducer is $b$-recurrent, then necessarily the
    unique state that can reach $q$ by reading $n$ letters $b$ backwards is
    itself, proving that $u\in L_{q'}$ and so $L_{p'} \subseteq L_{q'}$. Symmetrically,
    $L_{q'}\subseteq L_{p'}$, entailing $L_{p'}=L_{q'}$ and concluding the proof.
  \end{proofof}

  \begin{proofof}[Proof of Fact~\ref{fact:rectosub}]{Fact~\ref{fact:rectosub}}
    Let \(a\) be some letter and \(q\) be a state of the transducer.  Consider two
    states \(p, p'\) in \(q.a\).  Since \(q.a^\omega=\{q\}\), it holds that
    \(p.a^{\omega-1} = \{q\}\), and thus that \(p' \in p.a^\omega\).  Hence \(p =p'\), since
    \(p\) is \(a\)_recurrent.
  \end{proofof}

  \begin{proofof}[Proof of Fact~\ref{fact:subtov}]{Fact~\ref{fact:subtov} and
      Lemma~\ref{lem:nftocV}}
    Let \(\tau\) be a \(\cV\)_continuous utilitarian plurisubsequential
    transducer.  Consider an equation $u =_\cV v$, a state
    $q$ of $\tau$, and let $p = q.u$ and $p' = q.v$.  We show that $p=p'$,
    concluding this point.  We rely on the Syncing Lemma, since $\tau$ is
    $\cV$_continuous; it ensures in particular that:
    \begin{align}
      (\Tif{\tau, \nop, q} \sync \Tif{\tau, \nop, q}) \times
      (\Tif{\tau, p, \nop} \sync \Tif{\tau, p', \nop}) \subseteq
      \equ_\cV(u', v') \quad\text{with } u' = \Tif{\tau, q, p}(u), v' =
      \Tif{\tau, q, p'}(v)\enspace.\label{eqn:p}
    \end{align}
    Let $(s, s, t_1, t_2)$ be in the left-hand side.  We have that
    $s\cdot u' \cdot t_1 =_\cV s \cdot v' \cdot t_2$, thus
    $u' \cdot t_1 =_\cV v' \cdot t_2$ (here and in the following, we derive equivalent
    equations by appealing to the fact that the \emph{free group} is embedded,
    in a precise sense, in~$\cV$~\cite[\S~6.1.9]{robinson95}).  Now consider
    another tuple $(s', s', t_1', t_2')$ again in the left-hand side of
    Equation~(\ref{eqn:p}).  It also holds that
    $u' \cdot t_1' =_\cV v' \cdot t_2'$, hence we obtain that
    $t_1\cdot t_2^{-1} =_\cV t_1'\cdot t_2'^{-1}$.  This is in turn equal
    in~$\cV$ to some $\alpha\cdot \beta^{-1}$ such that $\alpha$ and $\beta$ are words that do not
    share the same last letter.  This shows that $t_1 = \alpha \cdot t$ and
    $t_2 = \beta \cdot t$ for some word $t$, and similarly for $t_1'$ and
    $t_2'$.  More generally:
    $(\Tif{\tau, p, \nop} \sync \Tif{\tau, p', \nop}) \subseteq (\alpha, \beta) \cdot \Id$, and since
    \(\tau\) is utilitarian, $p=p'$.
  \end{proofof}
\end{proofnoqed}

As an immediate corollary:
\begin{cor}\label{cor:grpconttotrans}
  For $\cV \in \{\vGnil, \vGsol, \vG\}$, any $\cV$_continuous rational function
  is $\cV$_realizable.
\end{cor}

\begin{thm}
  Let $\cV$ be a decidable variety of group languages that includes $\vGnil$ and
  that is closed under inverse $\cV$_realizable rational functions.  It is
  decidable, given an unambiguous transducer, whether it realizes a
  $\cV$_continuous function.  This holds in particular for $\vGsol$ and~$\vG$.
\end{thm}
\begin{proof}
  Lemma~\ref{lem:nftocV} together with Proposition~\ref{prop:transtocont} shows
  that a transducer is $\cV$_continuous iff its equivalent transducer
  effectively computed by Lemma~\ref{lem:nftocV} is a $\cV$_transducer.  This
  latter property being testable, the result follows.
\end{proof}

\subsection{Deciding continuity for aperiodic varieties}
We saw in Section~\ref{sec:struct} that the approach of the previous section
cannot work: there is no correspondence between continuity and realizability for
aperiodic varieties.  Herein, we use the Syncing Lemma to decide continuity in
two main steps.  First, note that all of our aperiodic varieties are defined by
an infinite number of equations for each alphabet.  The Syncing Lemma would thus
have us check an infinite number of conditions; our first step is to reduce this
to a finite number, which we stress through the forthcoming notion of
``pertaining triplet'' of states.  Second, we have to show that the inclusion of
the second point of the Syncing Lemma can effectively be checked.  This will be
done by simplifying this condition, and showing a decidability property on
rational relations.

We will need the following technical result in combinatorics on words in the
proof of the forthcoming Lemma~\ref{lem:apred}:
\begin{lem}\label{lem:uvxyst}
  Let $u, v, x, y, s, t \in A^*$ be words satisfying:
  \begin{enumerate}
  \item $u\cdot v,\; x\cdot y \in s^*$;
  \item $v\cdot u,\; y\cdot x \in t^*$;
  \item $s$ and $t$ are primitive.
  \end{enumerate}
  There exist $z, z' \in A^*$ such that:
  \begin{enumerate}
  \item $s = z\cdot z'$ and  $t = z' \cdot z$;
  \item $u,\; x \in s^*\cdot z$;
  \item $v,\; y \in t^*\cdot z'$.
  \end{enumerate}
\end{lem}
\begin{proof}
  Write $u = s^c\cdot z$ and $v = z'\cdot s^{c'}$ such that $s = z\cdot z'$.  It
  follows that $v\cdot u \in {(z'\cdot z)}^*$, and since $z' \cdot z$ is
  primitive, we have that $t = z' \cdot z$.  We can do the same with $x$ and
  $y$, letting $x = s^{\dot c} \cdot \dot z$ and
  $y = \dot z' \cdot s^{\dot c'}$, with $s = \dot z \cdot \dot z'$.  The same
  reasoning then shows that $t = \dot z' \cdot \dot z$.  Since $s$ has precisely
  $|s|$ conjugates (by~\cite[Proposition~1.3.2]{lothaire97}), we have that
  $\dot z = z$ and $\dot z' = z'$, and the properties of the Lemma follow.
\end{proof}

\begin{defi}\label{def:pert}
  A triplet of states $(p, q, q')$ is \emph{pertaining} if there are words $s,
  u, t$ and an integer $n$ such that:
  
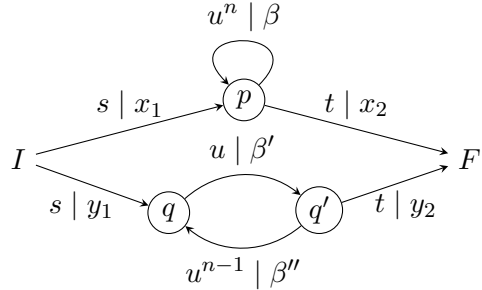
\begin{center}
    \begin{tikzpicture}[>=stealth,bend angle=40]
  \node (q0) {$I$};
  \node (p) at ($(q0.center)+(3,0.5)$)[anchor=south,state]  {$p$};
  \node (q) at ($(q0.center)+(2.,-1)$)[anchor=south,state ]  {$q$};
  \node (qp) at ($(q0.center)+(4.,-1)$)[anchor=south,state]  {$q'$};
  \node (F) at ($(q0.center)+(6,0.0)$)[]  {$F$};
  \draw[->] (q0) -- (p) node[pos=0.5, above] {$s\mid x_1$};
  \draw[->] (q0) -- (q) node[pos=0.4, below] {$s\mid y_1$};
  \draw[->] (p) -- (F) node[pos=0.5, above] {$t\mid x_2$};
  \draw[->] (qp) -- (F) node[pos=0.6, below] {$t\mid y_2$};
  \path[->] (q) edge[bend left] node[above] {$u\mid \beta'$} (qp) ;
  \path[->] (qp) edge[bend left] node[below] {$u^{n-1}\mid \beta''$}(q);
  \path[->] (p) edge[loop,looseness=7] node[above] {$u^n\mid \beta$} (p);
\end{tikzpicture}
  \end{center}
  where $\cdot$ means ``any word.'' Further, a pertaining triplet is
  \emph{empty} if, in the above picture, $\beta = \beta'\beta'' = 1$ and
  \emph{full} if both words are nonempty; it is \emph{degenerate} if only one of
  $\beta$ or $\beta'\beta''$ is empty.
\end{defi}

It is called ``pertaining'' as the second point of the Syncing Lemma elaborates
on properties of such a triplet, in particular, since $u^\omega = u^{\omega+1}$
is an equation of \vA.  The following characterization of $\vA$_continuity is
then made \emph{without appeal} to equations or profinite words:

\begin{lem}\label{lem:apred}
  A transducer $\tau\colon A^* \to B^*$ is $\vA$_continuous iff all of the
  following hold:
  \begin{enumerate}
  \item $\tau^{-1}(B^*) \in \vA(A^*)$;
  \item For all full pertaining triplets $(p, q, q')$, there exist
    $x, y \in B^*$ and $\rho_1, \rho_2 \in {(B^*)}^2$ such that
    $\Tif{\tau, \nop, p} \sync \Tif{\tau, \nop, q} \subseteq \Id\cdot \big((x^*,
    x^*)\rho_1^{-1}\big) \quad\text{and}\quad
    \Tif{\tau, p, \nop}
    \sync \Tif{\tau, q', \nop} \subseteq \big(\rho_2^{-1}(y^*, y^*)\big)\cdot
    \Id$;
  \item For all empty pertaining triplets $(p, q, q')$, we have that
    $(\Tif{\tau, \nop, p} \sync \Tif{\tau, \nop, q}) \cdot (\Tif{\tau, p, \nop}
    \sync \Tif{\tau, q', \nop}) \subseteq \Id$;
  \item No pertaining triplet is degenerate.
  \end{enumerate}
\end{lem}
\begin{proofof}{Lemma~\ref{lem:apred}}
  \proofstep{Only if}%
  Suppose $\tau$ is \vA_continuous, and let us appeal to the Syncing Lemma.
  Point 1 is then immediate.  Point 2 is a direct consequence of the second
  point of the Syncing Lemma and of Lemma~\ref{lem:apred}.

  We shall now check point 3, by contradiction.  Let $(p, q, q')$ be an empty
  pertaining triplet; we use the notations of Definition~\ref{def:pert}.  Then
  by functionality of $\tau$, we have that
  $\tau(s \cdot u^\omega \cdot t) = x_1 \cdot x_2$ and
  $\tau(s \cdot u^{\omega+1} \cdot t) = y_1 \cdot y_2$.  By the Preservation
  Lemma, and since
  $s \cdot u^\omega \cdot t =_\vA s \cdot u^{\omega+1} \cdot t$, it should hold
  that $x_1\cdot x_2 = y_1 \cdot y_2$, proving point 3.

  Point 4 is proven using similar ideas as point 3: with $(p, q, q')$ a
  degenerate pertaining tuple, and using the same notations as above, either the
  production of $s \cdot u^\omega \cdot t$ going through $p$ is a not a finite
  word while the production of $s \cdot u^{\omega+1} \cdot t$ through $q, q'$
  is, or vice-versa.  In both cases, it is not possible for these productions to
  be equal in $\vA$, hence if such a case happens, $\tau$ cannot be
  $\vA$_continuous.

  \proofstep{If}%
  We again rely on the Syncing Lemma, the first point of which being satisfied by
  hypothesis.  Let $u^\omega = u^{\omega+1}$ be an equation of $\vA$ with $u$ a
  word; the set of such equations defines $\vA(A^*)$.  Let $p, q, p', q'$ be
  states such that $p' \in p.u^\omega$ and $q' \in q.u^{\omega+1}$, and let
  $s, t$ be words with $p,q \in q_0.s$ and $p'.t, q'.t \in F$.  To conclude and
  apply the Syncing Lemma, we need to show that:
  \begin{align}
    \Tif{\tau, \nop, p}(s) \cdot \Tif{\tau, p, p'}(u^\omega) \cdot \Tif{\tau,
    p', \nop}(t) =_\vA \Tif{\tau, \nop, q}(s) \cdot \Tif{\tau, q,
    q'}(u^{\omega+1}) \cdot \Tif{\tau, q', \nop}(t)\enspace.\label{eqn:m}
  \end{align}
  (This is a direct consequence of the way profinite words are evaluated in a
  transducer, as per Lemma~\ref{lem:eval}.)

  Consider a large number $N = n!$, so that $p' \in p.u^N$ and $q' \in 
  q.u^{N+1}$.  With a large enough $N$, there must be two states $P$ and $Q$,
  and integers $i, j$ with $i + j = N$, such that:
  \begin{itemize}
  \item $P \in p.u^i$, $p' \in P.u^j$, and $P \in P.u^N$ (i.e., $P$ is
    ``between'' $p$ and $p'$, and belongs to a loop);
  \item $Q \in q.u^i$, $q' \in Q.u^{j+1}$, and $Q \in Q.u^N$.
  \end{itemize}
  (That such a pair exists can easily be seen on the product automaton of $\tau$
  by itself: The path from $(p, q)$ to $(p', q'')$ with $q' \in q''.u$ reading
  $u^N$ must go twice through the same pair of states $(P, Q)$, and this pair
  respects the above requirements.)

  Now define the following words:

  \def\bb#1#2{\makebox[#1][l]{#2}}
  \begin{itemize}
  \item \bb{2.7cm}{$\alpha = \Tif{\tau, \nop, P}(s\cdot u^i)$,}
    \bb{2.5cm}{$\beta = \Tif{\tau, P, P}(u^N)$,}
    $\gamma = \Tif{\tau, P, \nop}(u^j \cdot t)$,
  \item \bb{2.7cm}{$\alpha' = \Tif{\tau, \nop, Q}(s\cdot u^i)$,}
  \bb{2.5cm}{$\beta' = \Tif{\tau, Q, Q}(u^N)$,}
    $\gamma' = \Tif{\tau, Q, \nop}(u^{j+1} \cdot t)$.
  \end{itemize}
  Using the same reasoning as Lemma~\ref{lem:omega}, and the unambiguity of
  $\tau$, Equation (\ref{eqn:m}) is equivalent to:
  \[\alpha\cdot \beta^{\omega - N} \cdot \gamma =_\vA \alpha' \cdot
    \beta'^{\omega - N} \cdot \gamma'\enspace.\]
  Naturally, since $\alpha\cdot \beta^{\omega - N} \cdot \gamma =_\vA
  \alpha\cdot\beta^\omega\cdot\gamma$, and similarly for the right-hand side,
  Equation (\ref{eqn:m}) is equivalent to:
  \[\alpha\cdot \beta^\omega \cdot \gamma =_\vA \alpha' \cdot
    \beta'^\omega \cdot \gamma'\enspace.\]

  To make use of the hypotheses of the present Lemma, define $Q'$ to be in $Q.u$
  and such that $Q \in Q'.u^{N-1}$: $(P, Q, Q')$ is thus pertaining.  The
  situation is then:
  
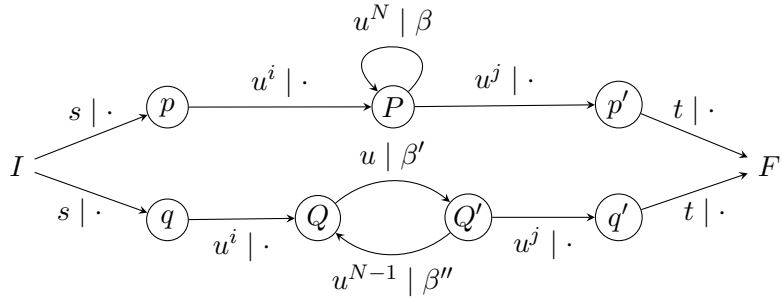
\begin{center}
    \begin{tikzpicture}[>=stealth,bend angle=40]
  \node (q0) {$I$};
  \node (p) at ($(q0.center)+(2,0.5)$)[anchor=south, state]  {$p$};
  \node (pp) at ($(q0.center)+(8,0.5)$)[anchor=south, state]  {$p'$};

  \node (q) at ($(q0.center)+(2.,-1)$)[anchor=south, state, ]  {$q$};
  \node (qp) at ($(q0.center)+(8.,-1)$)[anchor=south, state]  {$q'$};
  \node (P) at ($(q0.center)+(5,0.5)$)[anchor=south, state]  {$P$};
  \node (Q) at ($(q0.center)+(4.,-1)$)[anchor=south, state, ]  {$Q$};
  \node (Qp) at ($(q0.center)+(6.,-1)$)[anchor=south, state]  {$Q'$};
  \node (F) at ($(q0.center)+(10,0.0)$)[]  {$F$};
  \draw[->] (q0) -- (p) node[pos=0.5, above] {$s\mid\cdot$};
  \draw[->] (q0) -- (q) node[pos=0.4, below] {$s\mid\cdot$};
  \draw[->] (p) -- (P) node[pos=0.5, above] {$u^i\mid\cdot$};
  \draw[->] (Qp) -- (qp) node[pos=0.5, below] {$u^j\mid\cdot$};
  \draw[->] (q) -- (Q) node[pos=0.5, below] {$u^i\mid\cdot$};
  \draw[->] (P) -- (pp) node[pos=0.5, above] {$u^j\mid\cdot$};
  \draw[->] (pp) -- (F) node[pos=0.5, above] {$t\mid\cdot$};
  \draw[->] (qp) -- (F) node[pos=0.6, below] {$t\mid\cdot$};
  \path[->] (Q) edge[bend left] node[above] {$u\mid \beta'$} (Qp) ;
  \path[->] (Qp) edge[bend left] node[below] {$u^{N-1}\mid \beta''$}(Q);
  \path[->] (P) edge[loop,looseness=7] node[above] {$u^N\mid \beta$} (P);
\end{tikzpicture}
  \end{center}

  Since by hypothesis this triplet cannot be degenerate, either both of $\beta$ and
  $\beta'$ are empty, or none are.  Suppose they are both empty, then the
  hypothesis on empty triplets shows that:
  \[\Tif{\tau, \nop, P}(s \cdot u^i) \cdot \Tif{\tau, P, \nop}(u^{N+j} \cdot t) =
    \Tif{\tau, \nop, Q}(s \cdot u^i) \cdot \Tif{\tau, Q', \nop}(u^{N+j} \cdot
    t)\enspace.\]%
  The left-hand side evaluates to $\alpha \cdot \gamma$.  Since
  $\Tif{\tau, Q', \nop}(u^{N+j} \cdot t) = \Tif{\tau, Q', Q}(u^{N-1}) \cdot
  \Tif{\tau, Q, \nop}(u^{j+1} \cdot t) = \gamma'$, the right-hand side evaluates
  to $\alpha' \cdot \gamma'$, and Equation~(\ref{eqn:m}) is thus satisfied.

  Let us thus suppose that both $\beta$ and $\beta'$ are nonempty.  We divide
  $\beta'$ into $b_1b_2$ such that $b_1 = \Tif{\tau, Q, Q'}(u)$ and
  $b_2 = \Tif{\tau, Q', Q}(u^{N-1})$.  Now let $x, y \in B^*$ and
  $\rho_1, \rho_2 \in {(B^*)}^2$ be the (pairs of) words provided by point~2 for
  the triplet $(P, Q, Q')$.  Define
  $L = \Id \cdot \big((x^*, x^*)\rho_1^{-1}\big)$ and
  $R = \big(\rho_2^{-1}(y^*, y^*)\big) \cdot \Id$.  For any $k \geq 1$, and
  letting $\eta = s \cdot u^{i+k\times N}$ and
  $\eta' = u^{k \times N +j}\cdot t$, it holds by hypothesis that:
  \begin{itemize}
  \item
    $(\Tif{\tau, \nop, P}(\eta), \Tif{\tau, \nop, Q}(\eta)) = (\alpha \cdot
    \beta^k,\; \alpha' \cdot \beta'^k) \in L$;\hfill(a)
  \item
    $(\Tif{\tau, P, \nop}(\eta), \Tif{\tau, Q', \nop}(\eta)) = (\beta^k
    \cdot \gamma,\; b_2\cdot \beta'^{k-1}\cdot\gamma') \in R$.\hfill(b)
  \end{itemize}

  \noindent
  Let us first emphasize an easy property of $L$ and $R$:
  \begin{fact}\label{fact:pre}
    If $(w\cdot w', w \cdot w'') \in L$ with $|w'|, |w''| > |x|$, then
    $(w', w'') \in L$.  Moreover, if $(w, w') \in L$, then $w$ is a prefix of
    $w'$ or vice-versa.

    Similarly, if $(w'\cdot w, w'' \cdot w) \in R$ with $|w'|, |w''| > |y|$,
    then $(w', w'') \in R$.  Moreover, if $(w, w') \in R$, then $w$ is a suffix
    of $w'$ or vice-versa.
  \end{fact}
  \begin{proofof}{Fact~\ref{fact:pre}}
    We only show this for $L$, the case for $R$ being similar.

    For the first part of the statement, the hypothesis ensures the existence of
    a word $z$, integers $n', n''$, and two prefixes $x', x''$ of $x$ such that
    $w \cdot w' = z \cdot x^{n'} \cdot x'$ and
    $w \cdot w'' = z \cdot x^{n''} \cdot x''$.  If $w$ is a prefix of $z$, the
    property is easy to verify.  In the other cases, $w = z\cdot x^n\cdot \chi$
    for some integer $n < n', n''$ (strictness coming from the hypothesis) and
    $x = \chi\chi'$.  Hence $w' = \chi'\cdot x^{n'-n-1} \cdot x'$ and
    $w'' = \chi'\cdot x^{n''-n-1} \cdot x'$, and thus both belong to $L$.  The
    case of $R$ is similar.

    For the second part of the statement, $w$ and $w'$ start with a common word
    $z$, then some repetitions of $x$, and a prefix of $x$.  Clearly, one has to
    be a prefix of the other.
  \end{proofof}

  \def\subp{_\text{\normalfont p}}\def\subs{_\text{\normalfont s}}
  We first focus on the consequences of (a).  First, since either
  $\alpha\cdot\beta^k$ is a prefix of $\alpha'\cdot\beta'^k$ or vice-versa, we
  have that either $\alpha$ is a prefix of $\alpha'\cdot\beta'^k$, or $\alpha'$
  a prefix of $\alpha\cdot\beta^k$, for some $k$.  Suppose for instance that
  $\alpha' = \alpha\cdot\beta^c \cdot \beta\subp$, with
  $\beta = \beta\subp \cdot \beta\subs$; the other case will be treated later.
  Appealing to Fact~\ref{fact:pre}, for $k$ big enough, factoring out $\alpha'$
  yields that $({(\beta\subs\beta\subp)}^{k-c-1}\beta\subs, \beta'^k) \in L$.
  Hence ${(\beta\subs\beta\subp)}^*$ and $\beta'^*$ share common prefixes of
  unbounded length, implying that $\beta\subs\beta\subp$ and $\beta'$ are powers
  of a same primitive word $z_1$ (by, e.g.,~\cite[Proposition~1.3.5]{lothaire97}).

  Now similarly focusing on (b), we obtain that $\gamma$ is a suffix of
  $\beta'^k\cdot\gamma'$ or $\gamma'$ is a suffix of $\beta^k\cdot\gamma$, for
  some $k$.  Suppose for instance that
  $\gamma = \beta'\subs \cdot \beta'^{c'} \cdot \gamma'$, with
  $\beta' = \beta'\subp \cdot \beta'\subs$, again delaying the other case.  It
  follows, just as above, that $\beta$ and $\beta'\subs\beta'\subp$ are powers
  of a same primitive word $z_2$.  Noting that ${(\eta^c)}^\omega = \eta^\omega$,
  for any $\eta$, Equation~(\ref{eqn:m}) is thus equivalent to:
  \[\alpha \cdot z_2^\omega \cdot \beta'\subs\cdot\beta'^{c'} \cdot \gamma'
    =_\vA \alpha\cdot \beta^c \cdot   \beta\subp \cdot z_1^\omega \cdot
    \gamma'\enspace.\]

  Lemma~\ref{lem:uvxyst} indicates that there exist words $z, z'$ such that
  $z_1 = z\cdot z'$, $z_2 = z' \cdot z$, and
  $\beta'\subs \in z_2^*\cdot z', \beta\subp \in z' \cdot z_1^*$.  By
  eliminating $\alpha$ and $\gamma'$ we thus obtain that there are some integers
  $n_1, n_2$ such that Equation~(\ref{eqn:m}) is equivalent to
  $z_2^\omega\cdot z' \cdot z_1^{n_1 \times c'} =_\vA z_2^{n_2 \times c} \cdot
  z' \cdot z_1^\omega$, which clearly holds as both sides evaluate to
  ${(z' \cdot z)}^\omega \cdot z'$.

  \proofstep{Remaining cases}%
  We made two suppositions: $\alpha'$ is a prefix of $\alpha$, and $\gamma$ is a
  suffix of $\gamma'$.  The case where $\alpha$ is a prefix of $\alpha'$ and
  $\gamma'$ a suffix of $\gamma$ is entirely symmetric.  Let us keep our
  supposition on $\alpha'$ and assume that $\gamma'$ is a suffix of $\gamma$;
  the last remaining case is similar to this one.

  \def\betab{\dot{\beta}} Let us thus write
  $\gamma' = \betab\subs \cdot \beta^{c'} \cdot \gamma$, with
  $\beta = \betab\subp\cdot\betab\subs$.  We then obtain, factoring out
  $\gamma'$ this time, that $(\beta^{k-c-1}\betab\subp, \beta'^k) \in R$.  This
  implies that $(\betab\subs\betab\subp)$ and $\beta'$ are powers of the same
  primitive word, which can only be $z_1$.  Writing $z_2$ for the primitive root
  of $\beta$, Lemma~\ref{lem:uvxyst} shows the existence of words $z, z'$ such
  that $z_1 = z \cdot z'$, $z_2 = z' \cdot z$, and
  $\beta\subp \in z_2^*\cdot z', \betab\subs \in z \cdot z_2^*$.  By eliminating
  $\alpha$ and $\gamma$, we similarly obtain that Equation~(\ref{eqn:m}) is
  equivalent, for some $n_1, n_2$, to
  $z_2^\omega =_\vA z_2^{n_1 \times c} \cdot \beta\subp \cdot z_1^\omega \cdot
  \betab\subs \cdot z_2^{n_2 \times c'}$.  Then both sides evaluate to
  $z_2^\omega$, hence Equation~(\ref{eqn:m}) holds.
\end{proofof}

\begin{exa}
  We show that the transducer of Proposition~\ref{prop:conttrans} is
  $\vA$_continuous.  Let $\tau$ be:
  
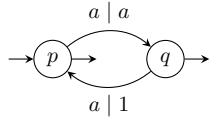
\begin{center}
    \begin{tikzpicture}[baseline={(current bounding box.center)},>=stealth,bend angle=40, scale=0.75,every node/.style={transform shape}]
  \node (P)  [draw,circle] {$p$};
  \node (Pin) at (P.west) [xshift=-1.5em] {};
  \node (Q) at ($(P.center)+(2.,0)$) [draw,circle]  {$q$};
  \node (Qout) at (Q.east) [xshift=1.5em] {};
  \node (Pout) at (P.east) [xshift=1.5em] {};
  \path[->] (P) edge[bend left] node[above] {$a\mid a$} (Q) ;
  \path[->] (Q) edge[bend left] node[below] {$a\mid 1$}(P);
  \path[->] (Pin) edge(P);
  \path[->] (P) edge (Pout);
  \path[->] (Q) edge (Qout);

\end{tikzpicture}
  \end{center}

  First, the function is total, hence the first point of Lemma~\ref{lem:apred}
  is satisfied.  Second, there are no empty nor degenerate pertaining triplets,
  hence the third and fourth points are satisfied.  Now the full pertaining
  triplets are $(p, p, p)$, $(p, p, q)$, $(q, q, q)$, and $(q, q, p)$.  We check
  that the pertaining triplet $(p, p, q)$ satisfies the second condition of
  Lemma~\ref{lem:apred}, the other cases being similar or clear.  The first half
  of the condition is immediate.  Now
  $\Tif{\tau,p,\nop} \sync \Tif{\tau,q,\nop} = \{(a^{\lfloor n+1/2\rfloor},
  a^{\lfloor n/2\rfloor}) \mid n \geq 0\}$ which satisfies the condition.
  \fullstop
\end{exa}

We now show that the property of Lemma~\ref{lem:apred} is indeed decidable:
\begin{prop}\label{prop:dectrans}
  It is decidable, given a rational relation $R \subseteq A^* \times A^*$,
  whether there is a word $x \in A^*$ and a pair $\rho \in
  {(A^*)}^2$, such that $R \subseteq \Id \cdot \big((x^*, x^*)\rho^{-1})$.
\end{prop}
\begin{proof}
  We rely on the classical result that it is decidable whether a rational
  relation is included in the identity~\cite[p.~650]{sakarovitch09}.

  We first tackle a related, simpler decision problem: Given a rational relation
  $R \subseteq (A^* \times A^*)$ and a word $x \in A^*$, check whether
  $R \subseteq \Id \cdot (x^*, x^*)$.  Write $f\colon A^* \to A^*$ for the
  function that removes the longest suffix in $x^*$ of its argument, and note
  that $f$ is a rational function.  Closure under inverse and composition of
  rational relations implies that $R' = \{(f(u), f(v)) \mid (u, v) \in R\}$ is a
  rational relation computable from $R$.  We have that $R' \subseteq \Id$ if and
  only if $R \subseteq \Id \cdot (x^*, x^*)$, hence the decision problem at hand
  is equivalent to checking whether $R' \subseteq \Id$, which is decidable.

  We now reduce the main decision problem to the previous one.


  First, we note that if a solution \((x, \rho)\) exists, then there is another
  solution \((\dot x, \dot \rho)\) with one component of \(\dot\rho\) empty.
  Indeed, write $x' = x\rho_1^{-1}, x''=x\rho_2^{-1}$.  Assume $x'$ is also a
  prefix of $x''$ (the symmetric case being similar).  We may thus write
  $x = x'y$ and $x'' = x'z$, and have that
  $R \subseteq \Id\cdot\big({(yx')}^*, {(yx')}^*z\big)$, showing that
  $\dot x := yx'$ and $\dot \rho := (1, z^{-1}yx')$ fit the requirements.

  We thus task ourselves with finding a solution \((x, \rho)\) with $\rho_2$
  empty (the symmetric case being similar).  We first check that
  $R \subseteq \Id$.  If this is not the case, we can compute a pair
  $(u, v) \in R \setminus \Id$ (again by~\cite[p.~650]{sakarovitch09}).  All the
  suffixes of $u$ are candidates for $x\rho_1^{-1}$; we go through all these
  candidates $x'$ (including the empty word).  We say that \(x'\) is a \emph{valid
    choice} if it stems from a valid solution \((x, \rho)\).

  Next, we verify that all pairs $(u, v) \in R$ are such that $u$ ends with $x'$
  (this is decidable, e.g., since
  $R \cap \big((A^*x')\comp \times A^*\big) = \emptyset$ is
  decidable~\cite[Proposition 2.6, Proposition 8.2]{berstel79}).  If this is not
  the case, then $x'$ is not a valid choice.  Otherwise, let us write
  $R' = R\cdot {(x',1)}^{-1}$, a rational relation.

  We now check again that $R' \subseteq \Id$; if it is the case, we are done and
  values of \((x, \rho)\) can be deduced.  Otherwise, we are given a pair
  $(u, v) \in R' \setminus \Id$.  Now if \(x'\) is a valid choice, then either
  $u$ is a prefix of $v$, or vice-versa.  In the former case, write
  $v = u\cdot z$; if \(x'\) is a valid choice, then $z \in x^*$, and this provides
  us with candidates for $x$: all the possible roots $z'$ of $z$.  We may now
  test that one such $z'$ starts with $x'$, and check whether
  $R \subseteq \Id\cdot(z'^*, z'^*)$ using the above decision problem.  If this
  holds, then there do exist an $x$ and a $\rho$ satisfying
  $R \subseteq \Id\cdot\big((x^*, x^*)\rho^{-1}\big)$.  Moreover, if such words
  exist, this procedure will find them.
\end{proof}

\begin{rem}\label{rk:guillon}
  In general, the problem of deciding, given a rational relation $R$ and a
  \emph{recognizable} relation $K$, whether $R \subseteq \Id \cdot K$, is
  undecidable.  Indeed, testing $R \cap \Id = \emptyset$ is
  undecidable~\cite{berstel79}, and equivalent to testing:
  \[R \subseteq \Id \cdot \big((A^+ \times \{1\}) \cup (\{1\} \times A^+) \cup
    \bigcup_{a \neq b \in A} (a\cdot A^* \times b\cdot A^*)\big)\enspace,\] the right-hand side
  being of the form $\Id \cdot K$.
  \fullstop
\end{rem}

\begin{thm}
  It is decidable, given an unambiguous transducer, whether it realizes an
  $\vA$_continuous function.
\end{thm}
\begin{proof}
  This is a consequence of Lemma~\ref{lem:apred}: Given a transducer, one can
  list all its pertaining triplets, and whether they are empty, full, or degenerate.
  For full pertaining triplets, the property of Lemma~\ref{lem:apred} is
  checked with Proposition~\ref{prop:dectrans} and the same Proposition applied
  on the reverse of the transducer.  The property for empty triplets can be
  checked since the inclusion of a rational relation in $\Id$ is decidable.
\end{proof}

The rest of this section focuses on conditions \emph{à la} Lemma~\ref{lem:apred}
for $\vJ$, $\vR$, $\vL$, and $\vDA$.  In each of these cases, we define the
proper notion of ``pertaining'' and rewrite the conditions of
Lemma~\ref{lem:apred} to match the defining equations.  Since the proofs are
simple variants of that of Lemma~\ref{lem:apred}, we omit them; we note that in
each case, the conditions are effectively verifiable.

\subsubsection{The case of $\vJ$}
We use a different set of equations to define $\vJ$, that can easily be proved
to be equivalent to the one given in the Preliminaries.  Specifically, $\vJ$ is
defined over any alphabet $A$ by the set of equations $x^\omega = y\cdot
x^\omega \cdot z$, with $y, z \in \cts(x)$.  The definition of ``pertaining''
then reads as follows:

\begin{defi}
  For two alphabets $C, D$, a quadruplet of states $(p, q, q', q'')$ is
  \emph{$(C, D)$_pertaining} if there are words $s, u, t$ with $\cts(u) = C$,
  words $z, z' \in C^*$, and an integer $n$ such that:
  
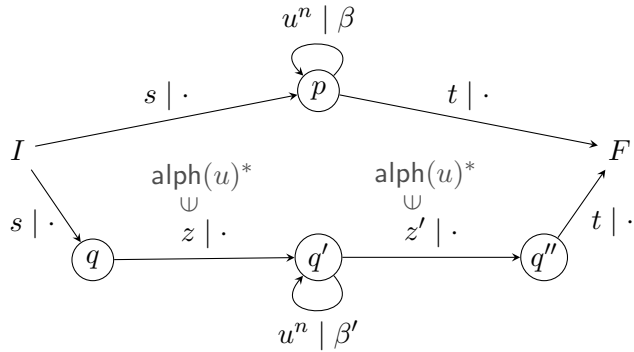
\begin{center}
    \begin{tikzpicture}[>=stealth,]

  \node (q0) {$I$};
  \node (p) at ($(q0.center)+(4,0.5)$)[anchor=south, state]  {$p$};
  \node (q) at ($(q0.center)+(1,-1.75)$)[anchor=south, state ]  {$q$};
  \node (qp) at ($(q0.center)+(4,-1.75)$)[anchor=south, state]  {$q'$};
  \node (qpp) at ($(q0.center)+(7,-1.75)$)[anchor=south, state]  {$q''$};
  \node (F) at ($(q0.center)+(8,0.0)$)[]  {$F$};
  \draw[->] (q0) -- (p) node[pos=0.5, above] {$s\mid\cdot$};
  \draw[->] (q0) -- (q) node[pos=0.4, below, xshift=-7pt] {$s\mid\cdot$};
  \draw[->] (p) -- (F) node[pos=0.5, above] {$t\mid\cdot$};
  \draw[->] (qpp) -- (F) node[pos=0.6, below,xshift=10pt] {$t\mid\cdot$} ;
  \path[->] (q) edge node (z1) [above] {$z \mid \cdot$} (qp) ;

  \path[->] (p) edge[loop,looseness=5] node [above] {$u^n\mid \beta $} (p);
 (q);
  \path[->] (qp) edge[loop below,looseness=5,out=-50, in=-130] node[below]
  {$u^n\mid \beta'$} (qp);
  \path[->] (qp) edge node (z2)[above] {$z'\mid  \cdot$} (qpp);
  \node (inz1) at (z1.north west) [ rotate=90, yshift=-7pt, xshift=2pt, color=black!70] {$\in$};
  \node (linz1) at (inz1.east) [ yshift=3pt,xshift=5pt,color=black!70] {$\cts(u)^*$};
  \node (inz2) at (z2.north west) [ rotate=90, yshift=-7pt, xshift=2pt, color=black!70] {$\in$};
  \node (linz2) at (inz2.east) [ yshift=3pt,xshift=5pt,color=black!70] {$\cts(u)^*$};
\end{tikzpicture}

  \end{center}
  and moreover $\cts(\beta) \cup \cts(\beta') = D$.  The pertaining quadruplet
  is \emph{empty} if $D = \emptyset$; it is \emph{full} if
  $\cts(\beta) = \cts(\beta') \neq \emptyset$, and \emph{degenerate} otherwise.
\end{defi}

\begin{lem}
  A transducer $\tau\colon A^* \to B^*$ is $\vJ$_continuous iff all of the
  following hold:
  \begin{enumerate}
  \item $\tau^{-1}(B^*) \in \vJ(A^*)$;
  \item For all full $(C, D)$_pertaining quadruplets $(p, q, q', q'')$:
    \begin{align*}
      \Tif{\tau, \nop, p} \sync \Big(\Tif{\tau, \nop, q} \cdot \big(\eps, \Tif{\tau,
      q, q'}(C^*)\big) \cdot \Tif{\tau, q', q'}\Big) & \subseteq \Id \cdot (D^*,
                                                       D^*)\enspace\text{and}\\
      \Tif{\tau, p, \nop} \sync \Big(\Tif{\tau, q', q'} \cdot \big(\eps,
      \Tif{\tau, q', q''}(C^*)\big) \cdot \Tif{\tau, q'', \nop}\Big) & \subseteq
                                                                       (D^*,
                                                                       D^*)\cdot \Id\enspace;
    \end{align*}
  \item For all empty $(C, D)$_pertaining quadruplets $(p, q, q', q'')$:
    \[(\Tif{\tau, \nop, p} \cdot \Tif{\tau, p, \nop}) \sync
      \Big(\Tif{\tau, \nop, q} \cdot \big(\eps, \Tif{\tau, q, q'}(C^*)\big)
      \cdot \Tif{\tau, q', q'} \cdot \big(\eps, \Tif{\tau, q', q''}(C^*)\big)
      \cdot \Tif{\tau, q'', \nop}\Big) \subseteq \Id\enspace;\]
  \item No pertaining quadruplet is degenerate.
  \end{enumerate}
\end{lem}

\subsubsection{The case of $\vR$}

Again, we slightly diverge from the usual equations for $\vR$, as presented in
the Preliminaries.  Indeed, $\vR$ is also defined, over any alphabet $A$, by
$x^\omega = x^\omega \cdot y$ with $y \in \cts(x)$.  We turn to the definition
of ``pertaining:''

\begin{defi}
  For an alphabet $C$, a triplet of states $(p, q, q')$ is
  \emph{$C$_pertaining} if there are words $s, u, t$ with $\cts(u) = C$,
  words $z \in C^*$, and an integer $n$ such that:
  
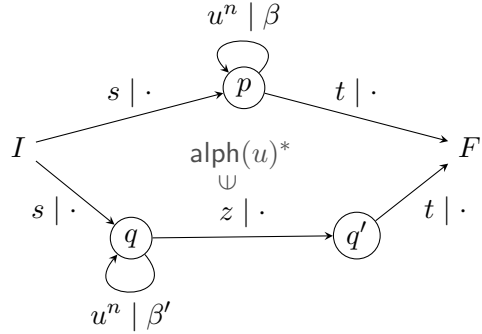
\begin{center}
    \begin{tikzpicture}[>=stealth,bend angle=40]
  \node (q0) {$I$};
  \node (p) at ($(q0.center)+(3,0.5)$)[anchor=south,state]  {$p$};
  \node (q) at ($(q0.center)+(1.5,-1.5)$)[anchor=south,state ]  {$q$};
  \node (qp) at ($(q0.center)+(4.5,-1.5)$)[anchor=south,state]  {$q'$};
  \node (F) at ($(q0.center)+(6,0.0)$)[]  {$F$};
  \draw[->] (q0) -- (p) node[pos=0.5, above] {$s\mid\cdot$};
  \draw[->] (q0) -- (q) node[pos=0.4, below, xshift=-5pt] {$s\mid\cdot$};
  \draw[->] (p) -- (F) node[pos=0.5, above] {$t\mid\cdot$};
  \draw[->] (qp) -- (F) node[pos=0.6, below,xshift=10pt] {$t\mid\cdot$};
  \path[->] (q) edge node (z1) [above] {$z\mid \cdot$} (qp) ;
    \node (inz1) at (z1.north west) [ rotate=90, yshift=-7pt, xshift=2pt, color=black!70] {$\in$};
  \node (linz1) at (inz1.east) [ yshift=3pt,xshift=5pt,color=black!70] {$\cts(u)^*$};
  \path[->] (p) edge[loop,looseness=5] node[above] {$u^n\mid \beta$} (p);
  \path[->] (q) edge[loop below, looseness=6,out=-50, in=-130] node[below] {$u^n\mid \beta'$} (q);

\end{tikzpicture}
  \end{center}
  The pertaining triplet is \emph{empty} if, in the above picture,
  $\beta = \beta' = 1$; it is \emph{full} if none of $\beta, \beta'$ is empty,
  and \emph{degenerate} otherwise.
\end{defi}

\begin{lem}
  A transducer $\tau\colon A^* \to B^*$ is $\vR$_continuous iff all of the
  following hold:
  \begin{enumerate}
  \item $\tau^{-1}(B^*) \in \vR(A^*)$;
  \item For all full $C$_pertaining triplets $(p, q, q')$, there exist $x \in
    B^*$ and $\rho \in {(B^*)}^2$ such that both inclusions hold:
    \begin{align*}
      \Tif{\tau, \nop, p} \sync \Tif{\tau, \nop, q} & \subseteq \Id\cdot
                                                      \big((x^*, x^*) \rho^{-1}\big)\enspace,\\
      \Tif{\tau, p, \nop}
      \sync \Big(\Tif{\tau, q, q} \cdot \big(\eps, \Tif{\tau, q, q'}(C^*)\big)
      \cdot \Tif{\tau, q', \nop}\Big) & \subseteq \big({\cts(x)}^*, {\cts(x)}^*\big)
                                        \cdot \Id\enspace;
    \end{align*}
  \item For all empty $C$_pertaining triplets $(p, q, q')$:
    \[(\Tif{\tau, \nop, p} \cdot \Tif{\tau, p, \nop}) \sync
      \Big(\Tif{\tau, \nop, q} \cdot \big(\eps, \Tif{\tau, q, q'}(C^*)\big)
      \cdot \Tif{\tau, q', \nop}\Big) \subseteq \Id\enspace;\]
  \item No pertaining triplet is degenerate.
  \end{enumerate}
\end{lem}

\noindent
Note that the $x$ of Proposition~\ref{prop:dectrans} can be effectively found.
The case of $\vL$ can be simply seen as the reversal of the previous case.

\subsubsection{The case of $\vDA$}

Similarly, we use a slightly less standard equational definition of $\vDA$.
Indeed, $\vDA$ is also defined, over any alphabet $A$, by $x^\omega = x^\omega
\cdot y \cdot x^\omega$ with $y \in \cts(x)$.  The definition of ``pertaining''
reflects these equations:
\begin{defi}
  For an alphabet $C$, a triplet of states $(p, q, q')$ is \emph{$C$_pertaining} if
  there are words $s, u, t$ with $\cts(u) = C$, a word $z \in C^*$, and an
  integer $n$ such that:

  
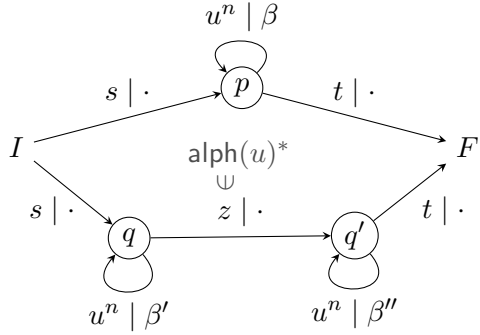
\begin{center}
    \begin{tikzpicture}[>=stealth,bend angle=40]
  \node (q0) {$I$};
  \node (p) at ($(q0.center)+(3,0.5)$)[anchor=south,state]  {$p$};
  \node (q) at ($(q0.center)+(1.5,-1.5)$)[anchor=south,state ]  {$q$};
  \node (qp) at ($(q0.center)+(4.5,-1.5)$)[anchor=south,state]  {$q'$};
  \node (F) at ($(q0.center)+(6,0.0)$)[]  {$F$};
  \draw[->] (q0) -- (p) node[pos=0.5, above] {$s\mid\cdot$};
  \draw[->] (q0) -- (q) node[pos=0.4, below, xshift=-5pt] {$s\mid\cdot$};
  \draw[->] (p) -- (F) node[pos=0.5, above] {$t\mid\cdot$};
  \draw[->] (qp) -- (F) node[pos=0.6, below,xshift=10pt] {$t\mid\cdot$};
  \path[->] (q) edge node (z1) [above] {$z\mid \cdot$} (qp) ;
    \node (inz1) at (z1.north west) [ rotate=90, yshift=-7pt, xshift=2pt, color=black!70] {$\in$};
  \node (linz1) at (inz1.east) [ yshift=3pt,xshift=5pt,color=black!70] {$\cts(u)^*$};
  \path[->] (p) edge[loop,looseness=5] node[above] {$u^n\mid \beta$} (p);
  \path[->] (q) edge[loop below, looseness=6,out=-50, in=-130] node[below] {$u^n\mid \beta'$} (q);
  \path[->] (qp) edge[loop below,looseness=5,out=-50, in=-130] node[below] {$u^n\mid \beta''$} (qp);

\end{tikzpicture}
  \end{center}

  Further, a pertaining triplet is \emph{empty} if, in the above picture,
  $\beta = \beta' = \beta'' = 1$; it is \emph{left-empty} if only $\beta'$ is
  empty, \emph{right-empty} if only $\beta''$ is empty, \emph{full} if none of
  $\beta, \beta', \beta''$ is empty, and \emph{degenerate} in the other cases.
\end{defi}

\begin{lem}
  A transducer $\tau\colon A^* \to B^*$ is $\vDA$_continuous iff all of the
  following hold:
  \begin{enumerate}
  \item $\tau^{-1}(B^*) \in \vDA(A^*)$;
  \item For all full $C$_pertaining triplets $(p, q, q')$, there exist
    $x, y \in B^*$ and $\rho_1, \rho_2 \in {(B^*)}^2$ such that these three
    inclusions hold:
    \begin{align*}
      \Tif{\tau, \nop, p} \sync \Tif{\tau, \nop, q} & \subseteq \Id\cdot \big((x^*,
                                                      x^*)\rho_1^{-1}\big)\enspace,\\
      \Tif{\tau, p, \nop}
      \sync \Tif{\tau, q', \nop} & \subseteq \big(\rho_2^{-1}(y^*, y^*)\big)\cdot
                                   \Id\enspace,\\
      \Tif{\tau, q, q'}(C^*) &\subseteq {\cts(x\cdot
                               y)}^*\enspace;
    \end{align*}
  \item For all empty $C$_pertaining triplets $(p, q, q')$:
    \[(\Tif{\tau, \nop, p} \sync \Tif{\tau, \nop, q}) \cdot (\eps, \Tif{\tau, q,
        q'}(C^*)) \cdot (\Tif{\tau, p, \nop} \sync \Tif{\tau, q', \nop}) \subseteq
      \Id\enspace;\]
  \item For all right-empty $C$_pertaining triplets $(p, q, q')$, there exist
    $x, y \in B^*$ and $\rho_1, \rho_2 \in {(B^*)}^2$ such that:
    \[\Tif{\tau, \nop, p} \sync \Tif{\tau, \nop, q} \subseteq \Id\cdot \big((x^*,
      x^*)\rho_1^{-1}\big) \quad\text{and}\quad
      \Tif{\tau, p, \nop} \sync \big((\eps, \Tif{\tau, q, q'}(C^*))\cdot
      \Tif{\tau, q', \nop}\big) \subseteq
      \big(\rho_2^{-1}(y^*, y^*)\big)\cdot \Id\enspace;\]
  \item For all left-empty $C$_pertaining triplets $(p, q, q')$, there exist
    $x, y \in B^*$ and $\rho_1, \rho_2 \in {(B^*)}^2$ such that:
    \[\Tif{\tau, \nop, p} \sync \big(\Tif{\tau, \nop, q} \cdot (\eps, \Tif{\tau, q, q'}(C^*))\big) \subseteq \Id\cdot \big((x^*,
      x^*)\rho_1^{-1}\big) \quad\text{and}\quad
      \Tif{\tau, p, \nop} \sync \Tif{\tau, q', \nop} \subseteq
      \big(\rho_2^{-1}(y^*, y^*)\big)\cdot \Id\enspace;\]
  \item No pertaining triplet is degenerate.
  \end{enumerate}
\end{lem}


\subsection[Deciding Com- and Ab-continuity]{Deciding Com- and Ab-continuity}

The case of $\vCom$ and $\vAb$ is comparatively much simpler, in particular
because these varieties are defined using a finite number of equations for each
alphabet.  However, the argument relies on different ideas:
\begin{thm}
  For $\cV = \vCom, \vAb$, it is decidable, given an unambiguous transducer,
  whether it realizes a $\cV$_continuous function.
\end{thm}
\begin{proof}
  We apply the Syncing Lemma.  Its first point is clearly decidable.  We reduce
  its second point to decidable properties about semilinear sets (see,
  e.g.,~\cite{ginsburg66}).  We also rely on the notion of Parikh image, that is,
  the mapping $\pkh\colon A^* \to \bbn^A$ such that $\pkh(w)$ maps $a \in A$ to
  the number of $a$'s in the word $w$.

  Since every $\vAb$_continuous function is $\vCom$_continuous
  (Proposition~\ref{prop:abtocom}), the conditions to test for $\vAb$_continuity
  are included in those for $\vCom$_continuity---this can also be seen as a
  consequence of the fact that if $u, v$ are words,
  $\equ_\vAb(u, v) = \equ_\vCom(u,v)$.

  Let $\tau\colon A^* \to B^*$ be a given transducer.  Consider an equation
  $ab = ba$ and four states $p, p', q, q'$ of $\tau$.  Write
  $u = \Tif{\tau, p, p'}(ab)$ and $v = \Tif{\tau, q, q'}(ba)$.  We ought to
  check, by the Syncing Lemma, the inclusion in
  $\equ_\vCom(u, v) = \{(s, s', t, t') \mid s\cdot u \cdot t =_\vCom s' \cdot v
  \cdot t'\}$ of some input synchronization.  Now this set is the set of
  $(s, s', t, t')$ such that
  $\pkh(s\cdot u \cdot t) = \pkh(s' \cdot v \cdot t')$, and is thus defined by a
  simple semilinear property.  The input synchronizations themselves, e.g.,
  $\Tif{\tau, \nop, p} \sync \Tif{\tau, \nop, q}$, are rational relations, and
  their component-wise Parikh image is thus a semilinear set.  Since the
  inclusion of semilinear sets is decidable, the inclusion of the second point
  of the Syncing Lemma is also decidable.

  For $\vAb$, we should additionally check the equations $a^\omega = 1$.  The
  reasoning is similar.  Consider three states $(p, p', q)$, and write
  $x \cdot u^{\omega-1} \cdot y$ for $\Tif{\tau, p, p'}(a^\omega)$.  By
  commutativity and the fact that $u^{\omega-1}$ acts as an inverse of $u$ in
  the equations holding in $\vAb$, we have that
  $(s, s', t, t') \in \equ_\vAb(x\cdot u^{\omega-1} \cdot y, 1)$ iff
  $s\cdot t =_\vAb s'\cdot u \cdot t'$.  This again reduces the inclusion of the
  second point of the Syncing Lemma to a decidable semilinear property.
\end{proof}

\section{Discussion}

We presented a study of continuity in functional transducers, on the one hand
focused on general statements (Section~\ref{sec:contapp}), on the other hand on
continuity for classical varieties.  The heart of this contribution resides in
decidability properties (Section~\ref{sec:dec}), although we also addressed
natural and related questions in a systematic way
(Section~\ref{sec:intermezzos}).  We single out two main research directions.

First, there is a sharp contrast between the genericity of the Preservation and
Syncing Lemma and the technicality of the actual proofs of decidability of
continuity.  To which extent can these be unified and generalized?  We know of
two immediate extensions: 1.\ the generic results of Section~\ref{sec:contapp}
readily apply to Boolean algebras of languages closed under quotient, a
relaxation of the conditions imposed on varieties, and 2.\
Proposition~\ref{prop:transtocont} and Lemma~\ref{lem:nftocV} can be shown to
also hold for the varieties $\vG_p$ of languages recognized by $p$-groups, hence
$\vGp$_continuity is decidable for transducers.  Beyond these two points, we do
not know how to show decidability for $\vGnil$ (which is the \emph{join} of the
$\vGp$), and the surprising complexity of the equalizer sets for some Burnside
varieties (e.g., the one defined by $x^2 = x^3$, see the Remark on
page~\pageref{rk:almeida}) leads us to conjecture that continuity may be
undecidable in that case, hence that no unified way to show the decidability of
continuity exists.

Second, the notion of continuity may be extended to more general settings.  For
instance, departing from regular languages, it can be noted that every recursive
function is continuous for the class of recursive languages.  Another natural
generalization consists in studying $(\cV, \cW)$_continuity, that is, the
property for a function to map $\cW$_languages to $\cV$_languages by inverse
image.  This would provide more flexibility for a sufficient condition for
cascades of languages (or stackings of circuits, or nestings of formulas) to be
in a given variety.

\subsubsection*{Acknowledgment.}  We are deeply indebted to Shaull Almagor,
Jorge Almeida (in particular for the Remark on page~\pageref{rk:almeida}), Luc
Dartois, Bruno Guillon (in particular for the Remark on
page~\pageref{rk:guillon}), Ismaël~Jecker, and Jean-Éric Pin for their
insightful comments and kind help.

The first and third authors are partly funded by the DFG Emmy Noether program
(KR~4042/2); the second author is funded by the DeLTA project
(ANR-16-CE40-0007). 

\bibliographystyle{plainurl}
\bibliography{bib}
\end{document}